\documentclass[12pt]{elsarticle}

\usepackage{graphicx}
\usepackage{amssymb}
\usepackage{epstopdf,comment}
\usepackage{amsthm,  mathrsfs, enumerate}

\usepackage{amsmath,color}
\usepackage{MnSymbol}

\newtheorem{proposition}{Proposition}
\newtheorem{theorem}{Theorem}
\newtheorem{corollary}{Corollary}
\newtheorem{lemma}{Lemma}
\newtheorem{assumption}{Assumption}
\newtheorem{remark}{Remark}

\numberwithin{theorem}{section}
\numberwithin{lemma}{section}
\numberwithin{equation}{section}
\numberwithin{proposition}{section}
\numberwithin{corollary}{section}

\newcommand{\req}[1]{Eq.\,(\ref{#1})}

\journal{Stochastic Processes and their Applications}

\begin{document} 
\begin{frontmatter}
\title{Homogenization of Dissipative, Noisy, Hamiltonian Dynamics}

\author{Jeremiah Birrell${}^a$, Jan Wehr${}^{a,b}$}
\address{${}^a$Department of Mathematics,\\
${}^b$Program in Applied Mathematics\\
University of Arizona\\
Tucson, AZ, 85721, USA}

\begin{abstract}
We study the dynamics of a class of Hamiltonian systems with dissipation, coupled to noise, in a singular (small mass) limit.  We derive the homogenized equation for the position degrees of freedom in the limit, including  the presence of a {\em noise-induced drift} term. We prove convergence to the solution of the homogenized equation in probability and, under stronger assumptions, in an $L^p$-norm.  Applications cover the overdamped limit of  particle motion in a time-dependent electromagnetic field, on a manifold with time-dependent metric, and the dynamics of nuclear matter.
\end{abstract}

\begin{keyword}
Hamiltonian system, homogenization, small mass limit, noise-induced drift\\
\MSC[2010] 60H10 \sep 82C31
\end{keyword}

\end{frontmatter}

\section{Introduction}
In the simplest case, the motion of  a diffusing particle of non-zero mass, $m$, is governed by a stochastic differential equation (SDE) of the form
\begin{align}\label{model_sys}
dq_t=v_t dt,\hspace{2mm} m dv_t=-\gamma v_t dt+\sigma dW_t,
\end{align}
where $\gamma$ and $\sigma$ are the dissipation (or drag) and diffusion coefficients respectively and $W_t$ is a Wiener process.  The study of diffusive systems in the limit $m\rightarrow 0$   was  initiated by Smoluchowski in \cite{smoluchowski1916drei} and continued by Kramers in \cite{KRAMERS1940284}. The field has grown to explore a large array of models and phenomena, including coupled fluid-particle systems \cite{doi:10.1137/S1540345903421076}, relativistic diffusion \cite{Chevalier2008,bailleul2010stochastic}, and a variety of processes and convergence modes on manifolds \cite{pinsky1976isotropic,pinsky1981homogenization,Jorgensen1978,dowell1980differentiable,XueMei2014,angst2015kinetic,bismut2005hypoelliptic,bismut2015}. History of the subject and a review of the early literature can be found in \cite{Nelson1967}. Such problems can be classified under the broad umbrella of homogenization, for which \cite{pavliotis2008multiscale} is an excellent reference. 

  Recently, there has been  increased interest in the phenomenon of {\em noise-induced drift}, which arises when the drag and noise coefficients are state dependent.  In such cases, the equation governing the process in the limit $m\rightarrow 0$ possesses an additional drift term that was not present in the original system.   First derived in \cite{PhysRevA.25.1130}, this has been observed experimentally in \cite{volpe2010influence} and  derived rigorously for one dimensional systems \cite{Sancho1982},  systems satisfying the fluctuation-dissipation relation \cite{PhysRevA.25.1130},  in Euclidean space of arbitrary dimension \cite{Hottovy2014,herzog2015small}, and on compact Riemannian manifolds of arbitrary dimension \cite{particle_manifold_paper}.  Further references to  work on the phenomenon of noise-induced drift are found in \cite{Hottovy2014}.

Statistical mechanics of fluctuating systems, as reviewed in \cite{Chetrite2008,gawedzki2013fluctuation}, covers systems more general than those governed by the Hamiltonians with quadratic kinetic energy,
\begin{align}\label{H_Newton}
H(q,p)=\frac{\|p\|^2}{2m}+V(q),
\end{align}
but to this point, the study of noise-induced drift has been restricted to Hamiltonians quadratic in $p$.  In this paper, we extend the theory to a large class of Hamiltonian systems generalizing \req{H_Newton}. See Section \ref{sec:examples} for examples of the type of systems that are covered. We prove that solutions to these more general Hamiltonian systems converge in probability and, under stronger assumptions, in an $L^p$-norm  to solutions of a homogenized limiting equation with a noise-induced drift term, for which we derive an explicit formula.  This is a far-reaching generalization of the previous results about the $m\rightarrow 0$ limit of the equations \req{model_sys}.

\subsection{Dissipative Hamiltonian System with Noise}
Here, we  review the basic equations and properties of dissipative, noisy Hamiltonian systems. See also  \cite{Chetrite2008}. Given a time-dependent Hamiltonian $H(t,x)$ which is $C^1$ jointly in $t\in\mathbb{R}$ and  $x=(q,p)\in \mathbb{R}^n\times \mathbb{R}^n$, a positive-semi-definite continuous matrix-valued function $\Gamma(t,x)$, the matrix
\begin{align}
 \Pi=\left( \begin{array}{cc}
0&I \\
-I & 0 \end{array} \right),
\end{align}
and a continuous vector field $G(t,x)$,  we first consider the following deterministic equation
\begin{align}\label{Hamiltonian_ODE}
\dot{x}_t=-\Gamma(t,x_t)\nabla H(t,x_t)+\Pi\nabla H(t,x_t)+G(t,x_t).
\end{align}
This equation describes the dynamics of a dissipative Hamiltonian system with drag matrix $\Gamma$ and external forcing $G$. We will refer to $q$ as the position degrees of freedom and to $p$ as the momentum degrees of freedom.

 The rate of change of the Hamiltonian along a solution is given by
\begin{align}
&\frac{d}{dt} H(t,x_t)\\
=&\partial_t H(t,x_t)-\nabla H(t,x_t)\cdot\Gamma(t,x_t)\nabla H(t,x_t)\notag\\
&+\nabla H(t,x_t)\cdot\Pi\nabla H(t,x_t)+\nabla H(t,x_t)\cdot G(t,x_t)\notag\\
\leq &\partial_t H(t,x_t)+\nabla H(t,x_t)\cdot G(t,x_t),\notag
\end{align}
where we used the anti-symmetry of $\Pi$ and the positive semi-definiteness of $\Gamma$. In particular, if $H$ is time independent and $G$ vanishes then the energy is non-increasing and if $\Gamma$ also vanishes then energy is conserved.  This justifies the interpretation of \req{Hamiltonian_ODE} as a dissipative Hamiltonian system with external forcing $G$ and drag matrix $\Gamma$.

We specialize to the case where the dissipation and external force enter only the momentum equation:
\begin{align}
\Gamma(t,x)=\left( \begin{array}{cc}
0&0 \\
0 & \gamma(t,x) \end{array} \right),
\end{align}
and $G(t,x)=(0,F(t,x))$. With this, the dissipation couples linearly to the generalized velocity $v=\nabla_p H$, since the equations are now
\begin{align}
\dot{q}_t=\nabla_p H(t,x_t),\hspace{2mm} \dot{p}_t=-\gamma(t,x_t)\nabla_p H(t,x_t)-\nabla_q H(t,x_t)+F(t,x_t).
\end{align}

We will be interested in families of Hamiltonians depending on some parameter $\epsilon\in (0,\epsilon_0]$ of the form
\begin{align}\label{H_family}
H^\epsilon(t,q,p)\equiv K^\epsilon(t,q,p)+V(t,q)\equiv K(\epsilon,t,q,(p-\psi(t,q))/\sqrt{\epsilon})+V(t,q),
\end{align}
where $V= V(t,q)$ is $C^2$, $K = K(\epsilon,t,q,z)$ is non-negative and $C^2$ in $(t,q,z)$ for each $\epsilon$, and $\psi$ is a $C^2$, $\mathbb{R}^n$-valued function.  
\begin{remark}
The momentum-dependent term, $K$, and the momentum-independent term, $V$, into which we split the Hamiltonian, do not have to carry with it the physical interpretation of kinetic and potential energy respectively, though we will use that terminology. The splitting will become constrained (though not quite unique) by further assumptions we will make below, but at this point it  is largely arbitrary.
\end{remark}

Families of Hamiltonians of the form \req{H_family} generalize the Hamiltonian of a classical particle coupled to an electromagnetic field (or in the case of vanishing vector potential, simply a Newtonian particle in a potential):
\begin{align}
H(t,q,p)=\frac{\|p-e\phi(t,q)\|^2}{2 m}+eV(t,q),
\end{align}
where $e$ is the charge, $\phi$ is the vector potential, and $V$ is the electrostatic potential. Scaling $p-e\phi(t,q)$ with $\sqrt{\epsilon}$ as in \req{H_family} is equivalent to replacing $m$ with $\epsilon m$, hence taking $\epsilon\rightarrow 0^+$ is equivalent to the small mass limit, $m\rightarrow 0$.   In this case $K(\epsilon,t,q,z)$ does not depend on $\epsilon$, but in general it is useful to allow an additional $\epsilon$ dependence. The form of this dependence will be somewhat constrained as we proceed.

 Adding a noise term to the momentum components of Hamilton's equations, we arrive at the following family of SDEs:\\
\begin{align}
dq^\epsilon_t=&\nabla_p H^\epsilon(t,x^\epsilon_t)dt,\label{Hamiltonian_SDE_q}\\
d p^\epsilon_t=&(-\gamma(t,x^\epsilon_t)\nabla_p H^\epsilon(t,x^\epsilon_t)-\nabla_q H^\epsilon(t,x^\epsilon_t)+F(t,x^\epsilon_t))dt+\sigma(t,x^\epsilon_t) dW_t,\label{Hamiltonian_SDE_p}
\end{align}
where $\sigma:[0,\infty)\times\mathbb{R}^n\rightarrow\mathbb{R}^{n\times k}$ is continuous and $W_t$ is a $\mathbb{R}^k$-valued Brownian motion on a filtered probability space $(\Omega,\mathcal{F},\mathcal{F}_t,P)$ satisfying the usual conditions \cite{karatzas2014brownian}.   In this paper we investigate the behavior of $x^\epsilon_t$ in the limit $\epsilon\rightarrow 0^+$ and derive a homogenized SDE satisfied by the limiting position process, $q_t$. 

\subsection{Summary of the Main Results}
To prove our first main theorem, Theorem \ref{conv_thm}, we will require several assumptions --- namely Assumptions  \ref{assump1}-\ref{assump7} in \ref{app:assump}.  Each one is restated in the body of the paper when it is first used, as not all are required for each result. These assumptions will constrain the initial conditions, the analytical properties and form of the Hamiltonian, the drag matrix, $\gamma$, and the noise coefficients, $\sigma$.  In particular, we will eventually require $\gamma$ to be independent of $p$ and its eigenvalues to satisfy a positive lower bound.  The latter coercivity requirement will be crucial in proving the kinetic energy and momentum bounds  in Section \ref{sec:K_bounds}. Under these assumptions we will prove the following:

Let $x_t^\epsilon$ be a family of solutions to the SDE \ref{Hamiltonian_SDE_q}-\ref{Hamiltonian_SDE_p} with  initial condition $x_0^\epsilon=(q_0^\epsilon,p_0^\epsilon)$. In this paper, we work under the assumption that a unique solution (pathwise uniqueness) exists for all $t\geq 0$ (i.e. there are no explosions).  See  \ref{app:no_explosions} for assumptions that guarantee this.

Then, as $\epsilon\to 0$, $q_t^\epsilon$ approaches the solution, $q_t$, to the SDE
\begin{align}\label{limit_eq_intro}
dq_t=&\tilde \gamma^{-1}(t,q_t)(-\partial_t\psi(t,q_t)-\nabla_{q}V(t,q_t)+F(t,q_t,\psi(t,q_t)))dt+S(t,q_t)dt\notag\\
&+\tilde\gamma^{-1}(t,q_t)\sigma(t,q_t,\psi(t,q_t)) dW_t.
\end{align}
The objects appearing in the SDE are defined as follows (here, and in the rest of the paper, we employ the summation convention on repeated indices):
\begin{enumerate}
\item $ \tilde\gamma_{ik}(t,q)\equiv\gamma_{ik}(t,q) +\partial_{q^k}\psi_i(t,q)-\partial_{q^i}\psi_k(t,q)$
\item $S^i(t,q)\equiv  Q^{ijl}(t,q)J_{jl}(t,q,\psi(t,q))$
\item $ Q^{ijl}(t,q)\equiv \partial_{q^k}(\tilde\gamma^{-1})^{ij}(t,q) A^{kl}(t,q)-\frac{1}{2}(\tilde\gamma^{-1})^{ik}(t,q)\partial_{q^k} A^{jl}(t,q)$, where $A^{ij}$ is the matrix-valued function from Assumption \ref{app:assump5} and the  index placement on $\tilde\gamma^{-1}$ is defined by $(\tilde\gamma^{-1})^{ij}\tilde\gamma_{jk}=\delta^i_k$.
\item $J_{ij}(t,x)\equiv G_{ij}^{kl}(t,q)\Sigma_{kl}(t,x)$
\item $G_{ij}^{kl}(t,q)\equiv\int_0^\infty (e^{-y (A\tilde\gamma)(t,q)})_i^k (e^{-y(A\tilde\gamma)(t,q)})_j^l dy$, where $A\tilde\gamma$ is the matrix $(A\tilde\gamma)^i_j=\tilde\gamma_{jk}A^{ki}$.
\item $\Sigma_{ij}\equiv\sum_\rho\sigma_{i\rho}\sigma_{j\rho}$
\end{enumerate}
It is interesting to note that the only feature of the kinetic energy function, $K$, that  plays a part in the limiting equation is the ``metric tensor" $A^{ij}$  from Assumption \ref{app:assump5}.

The convergence is in the following sense:\\
 Suppose that for all $\epsilon>0$ and  all $p>0$ we have $E[\|q^\epsilon_0\|^p]<\infty$, $E[\|q_0\|^p]<\infty$, and $E[\|q_0^\epsilon-q_0\|^p]=O(\epsilon^{p/2})$. Then for any  $T>0$, $p>0$, $0<\beta<p/2$ we have
\begin{align}\label{results_summary}
E\left[\sup_{t\in[0,T]}\|p_t^\epsilon-\psi(t,q^\epsilon_t)\|^p\right]=O(\epsilon^{\beta})\text{ and } E\left[\sup_{t\in[0,T]}\|q_t^\epsilon-q_t\|^p\right]=O(\epsilon^{\beta}) \text{ as } \epsilon\rightarrow 0^+.
\end{align}
We also prove a convergence in probability result,
\begin{align}\label{q_conv_rate_prob}
\lim_{\epsilon\to 0^+}P\left(\sup_{t\in[0,T]}\|q_t^\epsilon-q_t\|>\delta\right)=0 \text{ for all $T>0$, $\delta>0$, }
\end{align}
under less restrictive assumptions than the above $L^p$ result. See Theorem \ref{thm:conv_in_prob} for  details.

The  drift term, $S(t,q)$, that appears in the limiting equation is called the {\em noise-induced drift} and is nonzero when $\sigma$ is nonzero and a particular combination of $A$, $\gamma$, and $\psi$ have non-trivial state dependence.  Other works studying the small mass limit of inertial systems, both in Euclidean space \cite{Sancho1982,PhysRevA.25.1130,Hottovy2014,herzog2015small}, and on manifolds \cite{particle_manifold_paper}, have found analogous phenomena.

In addition to applying to a much larger class of Hamiltonians,  our derivation here gives a unified treatment of two previously studied systems: a particle in an electromagnetic field, \cite{Hottovy2014}, and a particle on a Riemannian manifold, \cite{particle_manifold_paper}.  These previous works used  different, and somewhat specialized, methods. Our results here also expand on these by allowing the metric tensor, forces, drag, and diffusion to be time-dependent. For Riemannian manifolds, we differ here by considering the non-compact case. See Section \ref{sec:examples} for details and further examples.

\subsection{Outline of the Proof}
The full details of the proof   begin in  Section \ref{sec:K_bounds}.  Here we outline our strategy and main ideas.

\begin{itemize}
\item In Section \ref{sec:K_bounds} we prove several results bounding the expectation of the kinetic energy.  For example, we show that for $T>0$, $q>0$ we have
\begin{align}
\sup_{t\in[0,T]}E[ K^\epsilon(t,x^\epsilon_t)^{q}]=O(1) \text{ as $\epsilon\rightarrow 0^+$.}
\end{align}
We use this to conclude several convergence results involving the momentum, for example
\begin{align}
\sup_{t\in[0,T]}E[\|p_t^\epsilon-\psi(t,q_t^\epsilon)\|^{q}]=O(\epsilon^{q/2}) \text{ as }\epsilon\rightarrow 0^+.
\end{align}

 The main tools for computing the estimates are It\^o's formula, several well known (stochastic) integral inequalities, and a lesser known $P$-a.s. stochastic integral inequality, Lemma \ref{matrix_exp_decay_bound}.

\item In Section \ref{sec:limit_eq} we derive the proposed form of the limiting equation for the position variables by solving \req{Hamiltonian_SDE_p} for $\nabla_p H^\epsilon(t,x^\epsilon_t)dt$, substituting into \req{Hamiltonian_SDE_q}, and integrating by parts to separate out the components that depend on  $p_t^\epsilon-\psi(t,q_t^\epsilon)$.  This is done with the aim of later using the  results of Section \ref{sec:K_bounds} to prove that these terms vanish in the limit.

The main complication here is that, in general, the required separation is possible only through formulating and solving an appropriate Lyapunov equation, \req{lyap_eq_def}.

\item In Section \ref{sec:conv_proof} we prove $L^p$-convergence of $q_t^\epsilon$ to the solution, $q_t$, of the proposed limiting equation from Section \ref{sec:limit_eq}.  This is accomplished by a  Gronwall's inequality argument.  The estimates of Section \ref{sec:K_bounds} are the critical ingredient here, allowing us to prove that the error terms converge to zero as $\epsilon\rightarrow 0^+$. This result relies on the assumption that the gradient of the potential is bounded (among others).

\item Finally, in Section \ref{sec:unbounded} we a technique adapted from \cite{herzog2015small} to prove  convergence in probability, \req{q_conv_rate_prob}, for a much wider range of systems, including many whose potentials have unbounded gradient.

The core idea is that, by modifying the objects in the SDE to be compactly supported (or at least have compactly supported derivatives) we can use the $L^p$ convergence result from Section \ref{sec:conv_proof} together with a limiting argument to prove convergence in probability.  A crucial ingredient is that none of the solutions (of the original or modified SDEs) explode in finite time.  This can be proven using the results of \ref{app:no_explosions}.

\end{itemize}

\section{Examples}\label{sec:examples}
Before we begin the proof, we first discuss several examples that fit within the above framework. 

\subsection{Particle in a Electromagnetic Field}
The Hamiltonian of a  particle of mass $\epsilon m$ and charge $e$ in an electromagnetic field with vector potential $\phi(t,q)$ and electrostatic potential $V(t,q)$ is
\begin{align}
H^\epsilon(t,q,p)=\frac{1}{2 \epsilon m}\|p-e \phi(t,q)\|^2+eV(t,q).
\end{align}
Allowing for an additional forcing, $F$, Hamilton's equations for this system are
\begin{align}
dq_t^\epsilon=&\frac{1}{\epsilon m}(p_t^\epsilon-e\phi(t,q_t^\epsilon))dt,\\
d(p_t^\epsilon)_i=&\left(-\frac{1}{\epsilon m}\gamma_{ij}(t,q_t^\epsilon)\delta^{jk}((p_t^\epsilon)_k-e\phi_k(t,q_t^\epsilon))+F_i(t,x_t^\epsilon)-e\partial_{q^i} V(t,q_t^\epsilon)\right.\\
&\left.+\frac{e}{\epsilon m}\partial_{q^i}\phi_k(t,q_t^\epsilon)\delta^{jk}((p_t^\epsilon)_j-e\phi_j(t,q_t^\epsilon))\right)dt+\sigma_{i\rho}(t,x_t^\epsilon)dW^\rho_t.\notag
\end{align}
The homogenized equation in the small mass limit is difficult to simplify further than \req{limit_eq_intro} in general. However, in the case where $\gamma$ and $\sigma$ are independent of $p$ and the fluctuation dissipation relation holds pointwise for a time and position dependent ``temperature" $T(t,q)$,
\begin{align}
\Sigma_{ij}(t,q)=2k_BT(t,q) \gamma_{ij}(t,q),
\end{align}
one can show that
\begin{align}
 G_{kl}^{ab}(t,q)\Sigma_{ab}(t,q)=k_BT(t,q)\delta_{kl},
\end{align}
where $G$ was defined in \req{G_def},  $\Sigma_{ij}=\sum_\rho\sigma_{i\rho}\sigma_{j\rho}$, and $k_B$ is Boltzmann's constant.

The noise induced drift, \req{noise_induced_drift}, can therefore be simplified to
\begin{align}
S^i(t,q)=k_BT(t,q)\partial_{q^j}(\tilde\gamma^{-1})^{ij}(t,q_t).
\end{align}
Recall that we defined
\begin{align}
\tilde\gamma_{ik}(t,q)\equiv\gamma_{ik}(t,q) +\partial_{q^k}\psi_i(t,q)-\partial_{q^i}\psi_k(t,q),
\end{align}
where here, $\psi=e\phi$.

The homogenized equation in the small mass limit is then
\begin{align}
dq_t^i=&(\tilde\gamma^{-1})^{ij}(t,q_t)(-\partial_t\psi_j(t,q_t)-e\partial_{q^j}V(t,q_t)+F_j(t,q_t,\psi(t,q_t)))dt\\
&+k_BT(t,q)\partial_{q^j}(\tilde\gamma^{-1})^{ij}(t,q_t)dt+(\tilde\gamma^{-1})^{ij}(t,q_t)\sigma_{j\rho}(t,q_t)dW^\rho_t.\notag
\end{align}
The time independent case was studied in \cite{Hottovy2014} by a different method and coincides with the results in this paper.

\subsection{Particle on a Riemannian Manifold}
Another case that is covered by the framework developed here is the inertial motion of a particle in $\mathbb{R}^n$, but with geometry specified by a time-dependent Riemannian metric tensor, $g_{ij}(t,q)$.  The family of Hamiltonians describing this system is
\begin{align}
H^\epsilon(t,q,p)=\frac{1}{2 \epsilon m}g^{ij}(t,q)p_ip_j.
\end{align}
 Note that the inverse metric tensor, $g^{ij}(t,q)$, is playing the role of $A^{ij}(t,q)$ in our formalism, and so all the assumptions that are required of $A^{ij}$ there must be satisfied by $g^{ij}$ here.

Allowing for external forcing, $F$, Hamilton's equations are
\begin{align}
d(q_t^\epsilon)^i=&\frac{1}{\epsilon m} g^{ij}(t,q_t^\epsilon)(p_t^\epsilon)_jdt,\\
d(p_t^\epsilon)_i=&\left(-\frac{1}{\epsilon m}\gamma_{ij}(t,x^\epsilon_t) g^{jk}(t,q_t^\epsilon)(p_t^\epsilon)_k-\frac{1}{2\epsilon m}\partial_{q^i}g^{kl}(t,q_t^\epsilon)(p_t^\epsilon)_k(p_t^\epsilon)_l+F_i(t,x^\epsilon_t)\right)dt\notag\\
&+\sigma_{i\rho}(t,x^\epsilon_t) dW^\rho_t.
\end{align}

Again, the homogenized equation in the small mass limit, \req{limit_eq1}, can be simplified if $\gamma$ and $\sigma$ are independent of $p$ and the fluctuation dissipation relation holds pointwise for a time and position dependent ``temperature" $T(t,q)$,
\begin{align}
\Sigma_{ij}(t,q)=2k_BT(t,q) \gamma_{ij}(t,q).
\end{align}

In this case one finds that 
\begin{align}
G_{kl}^{ab}(t,q)\Sigma_{ab}(t,q)=k_BT(t,q)g_{kl}(t,q)
\end{align}
 and hence \req{limit_eq1} becomes
\begin{align}
dq_t^i=&(\gamma^{-1})^{ij}(t,q_t)F_j(t,q_t,0)dt+S^i(t,q_t)dt+(\gamma^{-1})^{ij}(t,q_t)\sigma_{j\rho}(t,q_t)dW^\rho_t,
\end{align}
where the noise induced drift is 
\begin{align}
S^i(t,q)=&k_BT(t,q)\left( \partial_{q^j}(\gamma^{-1})^{ij}(t,q) -\frac{1}{2}(\gamma^{-1})^{ij}(t,q)g_{kl}(t,q)\partial_{q^j} g^{kl}(t,q)\right).
\end{align}
See also \cite{particle_manifold_paper}, which treats the case of a smooth, compact, connected, manifold without boundary (but otherwise arbitrary topology) and with time independent metric via a more geometrically motivated approach.  

One can argue that the present approach is simpler, as we only rely on tools from (stochastic) analysis; we avoid the geometrical machinery used in \cite{particle_manifold_paper}.  In addition, here we don't require compactness; on the other hand, we do lose the ability to handle non-trivial topology of the manifold.    

\subsection{Hamiltonian that are Polynomials in the Momentum}
Generalizing the above two quadratic cases, our convergence result applies to Hamiltonians that are polynomials in $p-\psi(t,q)$.  Specifically,   Theorem \ref{conv_thm} implies $L^p$ convergence if the family of Hamiltonians has the form
\begin{align}
H^\epsilon(t,q,p)=\sum_{l=k_1}^{k_2} d_l(t) \left[A^{ij}(t,q) (p-\psi(t,q))_i(p-\psi(t,q))_j/\epsilon\right]^l+V(t,q)
\end{align}
where $1\leq k_1\leq k_2$ are integers  and the following properties hold on $[0,T]\times\mathbb{R}^{n}$ for every $T>0$:
\begin{enumerate}
\item $V$ is $C^2$ and $\nabla_q V$ is bounded and Lipschitz in $q$, uniformly in $t$.
\item $\psi$ is $C^3$ and $\partial_t\psi$, $\partial_{q^i}\psi$,  $\partial_{q^i}\partial_{q^j}\psi$, $\partial_t\partial_{q^i}\psi$, $\partial_t\partial_{q^j}\partial_{q^i}\psi$, and $\partial_{q^l}\partial_{q^j}\partial_{q^i}\psi$ are bounded.
\item $d_l$ are $C^2$ and non-negative.
\item $d_{k_1}$ and $d_{k_2}$ are uniformly bounded below by a positive constant.
\item $A$ is $C^2$, positive-definite, and $A$, $\partial_t A$, $\partial_{q^i} A$,  $\partial_t \partial_{q^i}A$,  and $\partial_{q^i}\partial_{q^j} A$  are bounded.
\item The eigenvalues of $A$ are uniformly bounded below by a positive constant.
\end{enumerate}
As mentioned above, convergence in probability holds under much weaker assumptions on the potential. See Theorem \ref{thm:conv_in_prob}.

\subsection{Effective Nuclear Interactions}
Mean field models of nuclear interactions can lead to  non-quadratic momentum dependence in the Hamiltonian.  For example, in \cite{KHOA1992102} a contribution to the effective potential of the form
\begin{align}
U(p)= c_1\ln^2\left[1+ c_2\|p\|^2\right]
\end{align}
was calculated.  For non-relativistic particle motion, a term of this form can be accommodated in our framework in several ways, depending on which parameters one wishes to scale. For example, one can let
\begin{align}
K^\epsilon(p)=\frac{\|p\|^2}{2\epsilon m}+ c_1\ln^2\left[1+  c_2\|p\|^2\right]
\end{align}
or
\begin{align}
K^\epsilon(p)=\frac{\|p\|^2}{2\epsilon m}+ c_1\ln^2\left[1+  c_2\|p\|^2/\epsilon\right].
\end{align}
In either case, the assumptions of Theorem \ref{conv_thm} relating to $K$ are satisfied.

\section{Kinetic Energy and Momentum Bounds}\label{sec:K_bounds}

We now begin working towards the proof of our main results. In this section,  we derive bounds on the behavior of the kinetic energy  in the limit $\epsilon\rightarrow 0^+$.  As a consequence we will obtain a convergence result for the canonical momentum, the first formula in \req{results_summary}.  Some assumptions on the structure of the Hamiltonian are required.  As usual, here and in the sequel, generic symbols, denoting constants, such as $C$, $M$ etc., do not have to have the same value in all equations.  
\begin{assumption}\label{assump1}
We assume that $\sigma$, $F$, and $\gamma$ are continuous, the Hamiltonian has the form given in \req{H_family} where $K(\epsilon,t,q,z)$ is non-negative and $C^2$ in $(t,q,z)$ for each $\epsilon$, $\psi$ is  $C^2$, and the solutions, $x^\epsilon_t$, to the SDE \ref{Hamiltonian_SDE_q}-\ref{Hamiltonian_SDE_p} exist for all $t\geq 0$.

  For every $T>0$, we assume the following bounds hold on $(0,\epsilon_0]\times [0,T]\times\mathbb{R}^{2n}$:
\begin{enumerate}
\item There exist $C>0$ and $M>0$ such that
\begin{align}\label{K_assump1}
\max\{|\partial_t K(\epsilon,t,q,z)|,\|\nabla_q K(\epsilon,t,q,z)\|\}\leq M+CK(\epsilon,t,q,z).
\end{align}
\item There exist  $c>0$ and $M\geq 0$ such that
\begin{align}\label{K_assump2}
\|\nabla_z K(\epsilon,t,q,z)\|^2+M\geq c K(\epsilon,t,q,z).
\end{align}
\item
For every $\delta>0$ there exists an $M>0$ such that
\begin{align}\label{K_assump3}
\max\left\{\|\nabla_z K(\epsilon,t,q,z)\|,\left(\sum_{i,j}|\partial_{z_i}\partial_{z_j}K(\epsilon,t,q,z)|^2\right)^{1/2}\right\}\leq M+\delta K(\epsilon,t,q,z).
\end{align}
\end{enumerate}

\end{assumption}

We will also need the following assumptions concerning the potential, dissipation, noise, external forcing, and initial conditions.  Some of these will be relaxed in Section \ref{sec:unbounded}.
\begin{assumption}\label{assump2}
For every $T>0$, we assume that the following hold uniformly on $[0,T]\times\mathbb{R}^n$:
\begin{enumerate}
\item $V$ is $C^2$ and $\nabla_q V$ is bounded.
\item $\gamma$ is symmetric with eigenvalues  bounded below by some $\lambda>0$.
\item  $\gamma$, $F$, $\partial_t\psi$, and $\sigma$ are bounded.
\item There exists $C>0$ such that the (random) initial conditions satisfy $K^\epsilon(0,x^\epsilon_0)\leq C$ for all $\epsilon>0$ and all $\omega\in\Omega$.
\end{enumerate}
\end{assumption}

We now state and prove the kinetic energy bound which underlies our main results. As with our notation $K^\epsilon$, for any function $f(\epsilon,t,q,z)$ we  define
\begin{align}
f^\epsilon(t,x)\equiv f(\epsilon,t,q,(p-\psi(t,q))/\sqrt{\epsilon}).
\end{align}
For example,
\begin{align}
(\partial_{z^i}K)^\epsilon(t,x)\equiv \partial_{z_i}K(\epsilon,t,q,(p-\psi(t,q))/\sqrt{\epsilon})
\end{align}
and similarly for $(\nabla_zK)^\epsilon$, $(\partial_{z^i}\partial_{z^j}K)^\epsilon$, etc. 
\begin{lemma}\label{K_bound_lemma}
Under Assumptions \ref{assump1} and \ref{assump2}, for any $q\in\mathbb{N}$, $q\geq 1$ and any $T>0$ there exist $\alpha_0>0$, $\epsilon_0>0$ and $\kappa>0$ such that for all $0<\alpha\leq \alpha_0$, $0<\epsilon\leq \epsilon_0$, $0\leq t\leq T$ we have the $P$-a.s. inequality
\begin{align}\label{K_bound_lemma_eq}
 &K^\epsilon(t,x_t^\epsilon)^q\leq \frac{\kappa}{\alpha}+\frac{q}{\sqrt{\epsilon}} e^{-\alpha t/\epsilon} \int_0^t e^{\alpha s/\epsilon}K^\epsilon(s,x_s^\epsilon)^{q-1}(\nabla_{z} K)^\epsilon(s,x_s^\epsilon)\cdot\sigma(s,x^\epsilon_s) dW_s.
\end{align}
\end{lemma}
\begin{proof}
Take $T>0$, $q\in\mathbb{N}$, $q\geq 1$, $\alpha>0$, and apply It\^o's formula to $e^{\alpha t/\epsilon}K^\epsilon(t,x_t^\epsilon)^q$:
\begin{align}
&e^{\alpha t/\epsilon}K^\epsilon(t,x_t^\epsilon)^q=K^\epsilon(0,x_0^\epsilon)^q+\frac{\alpha}{\epsilon}\int_0^te^{\alpha s/\epsilon}K^\epsilon(s,x_s^\epsilon)^qds\\
&+q\int_0^te^{\alpha s/\epsilon}K^\epsilon(s,x_s^\epsilon)^{q-1}(\partial_s K)^\epsilon(s,x_s^\epsilon)ds\notag\\
&-\frac{q}{\epsilon}\int_0^te^{\alpha s/\epsilon}K^\epsilon(s,x_s^\epsilon)^{q-1}(\nabla_z K)^\epsilon(s,x_s^\epsilon)\cdot\gamma(s,x^\epsilon_s)(\nabla_z K)^\epsilon(s,x_s^\epsilon)ds\notag\\
&+\frac{q}{\sqrt{\epsilon}}\int_0^te^{\alpha s/\epsilon}K^\epsilon(s,x_s^\epsilon)^{q-1}(\nabla_z K)^\epsilon(s,x_s^\epsilon)\cdot(-\partial_s\psi(s,q^\epsilon_s)-\nabla_q V(s,q^\epsilon_s)+F(s,x^\epsilon_s))ds\notag\\
&+\frac{q(q-1)}{2\epsilon}\int_0^te^{\alpha s/\epsilon}K^\epsilon(s,x_s^\epsilon)^{q-2}(\nabla_z  K)^\epsilon(s,x_s^\epsilon)\cdot \Sigma(s,x_s^\epsilon) (\nabla_{z} K)^\epsilon(s,x_s^\epsilon) ds\notag\\
&+\frac{q}{2\epsilon}\int_0^te^{\alpha s/\epsilon}K^\epsilon(s,x_s^\epsilon)^{q-1}(\partial_{z_i}\partial_{z_j} K)^\epsilon(s,x_s^\epsilon)\Sigma_{ij}(s,x_s^\epsilon) ds\notag\\
&+\frac{q}{\sqrt{\epsilon}}\int_0^te^{\alpha s/\epsilon}K^\epsilon(s,x_s^\epsilon)^{q-1}(\nabla_{z} K)^\epsilon(s,x_s^\epsilon)\cdot\sigma(s,x^\epsilon_s) dW_s\notag.
\end{align}
Here we used the fact that  $\nabla_pH^\epsilon=\nabla_p K^\epsilon$ to cancel the terms involving $\nabla_q K^\epsilon$.

Using Assumption  \ref{assump2}, for any $t\in[0,T]$ we have
\begin{align}
&e^{\alpha t/\epsilon}K^\epsilon(t,x_t^\epsilon)^q\\
\leq&K^\epsilon(0,x_0^\epsilon)^q+\frac{\alpha}{\epsilon}\int_0^te^{\alpha s/\epsilon}K^\epsilon(s,x_s^\epsilon)^qds\notag\\
&+q\int_0^te^{\alpha s/\epsilon}K^\epsilon(s,x_s^\epsilon)^{q-1}(\partial_s K)^\epsilon(s,x_s^\epsilon)ds\notag\\
&-\frac{q\lambda}{\epsilon}\int_0^te^{\alpha s/\epsilon}K^\epsilon(s,x_s^\epsilon)^{q-1}\|(\nabla_z K)^\epsilon(s,x_s^\epsilon)\|^2ds\notag\\
&+\frac{q}{\sqrt{\epsilon}}\|-\partial_t\psi-\nabla_q V+F\|_\infty\int_0^te^{\alpha s/\epsilon}K^\epsilon(s,x_s^\epsilon)^{q-1}\|(\nabla_z K)^\epsilon(s,x_s^\epsilon)\|ds\notag\\
&+\frac{q(q-1)}{2\epsilon}\|\Sigma\|_\infty\int_0^te^{\alpha s/\epsilon}K^\epsilon(s,x_s^\epsilon)^{q-2}\|(\nabla_{z} K)^\epsilon(s,x_s^\epsilon)\|^2 ds\notag\\
&+\frac{q}{2\epsilon}\|\Sigma\|_{F,\infty}\int_0^te^{\alpha s/\epsilon}K^\epsilon(s,x_s^\epsilon)^{q-1}\left(\sum_{i,j}(\partial_{z_i}\partial_{z_j} K)^\epsilon(s,x_s^\epsilon)^2\right)^{1/2}ds\notag\\
&+\frac{q}{\sqrt{\epsilon}}\int_0^te^{\alpha s/\epsilon}K^\epsilon(s,x_s^\epsilon)^{q-1}(\nabla_{z} K)^\epsilon(s,x_s^\epsilon)\cdot\sigma(s,x^\epsilon_s) dW_s.\notag
\end{align}
 Here and in the following, $\|Y\|_F$ will denote the Frobenius (or Hilbert-Schmidt) norm of a matrix $Y$, i.e. $\|Y\|_F= \left(\sum_{i,j}Y_{ij}^2\right)^{1 \over 2}$.  For any matrix or vector-valued quantity $Y$ we write $\|Y\|_{\infty}\equiv\sup_{(t,x)\in[0,T]\times\mathbb{R}^{2n}}\|Y(t,x)\|$ and similarly for $\|\cdot\|_{F,\infty}$.  The implied value of $T$ will be clear from the context.

For any $\delta>0$,  Assumption \ref{assump1} implies the existence of $C>0$, $c>0$, and $M>0$ such that
\begin{align}
&e^{\alpha t/\epsilon}K^\epsilon(t,x_t^\epsilon)^q\\
\leq&K^\epsilon(0,x_0^\epsilon)^q+\frac{\alpha}{\epsilon}\int_0^te^{\alpha s/\epsilon}K^\epsilon(s,x_s^\epsilon)^qds\notag\\
&+q\int_0^te^{\alpha s/\epsilon}K^\epsilon(s,x_s^\epsilon)^{q-1}(M+C K^\epsilon(s,x_s^\epsilon))ds\notag\\
&-\frac{q\lambda}{\epsilon}\int_0^te^{\alpha s/\epsilon}K^\epsilon(s,x_s^\epsilon)^{q-1}(c K^\epsilon(s,x_s^\epsilon)-M)ds\notag\\
&+\frac{q}{\sqrt{\epsilon}}\|-\partial_t\psi-\nabla_q V+F\|_\infty\int_0^te^{\alpha s/\epsilon}K^\epsilon(s,x_s^\epsilon)^{q-1}(M+\delta K^\epsilon(s,x_s^\epsilon))ds\notag\\
&+\frac{q(q-1)}{2\epsilon}\|\Sigma\|_\infty\int_0^te^{\alpha s/\epsilon}K^\epsilon(s,x_s^\epsilon)^{q-2}(M+\delta K^\epsilon(s,x_s^\epsilon))^2 ds\notag\\
&+\frac{q}{2\epsilon}\|\Sigma\|_{F,\infty}\int_0^te^{\alpha s/\epsilon}K^\epsilon(s,x_s^\epsilon)^{q-1}(M+\delta K^\epsilon(s,x_s^\epsilon))ds\notag\\
&+\frac{q}{\sqrt{\epsilon}}\int_0^te^{\alpha s/\epsilon}K^\epsilon(s,x_s^\epsilon)^{q-1}(\nabla_{z} K)^\epsilon(s,x_s^\epsilon)\cdot\sigma(s,x^\epsilon_s) dW_s.\notag
\end{align}
In the estimate that follows, the first two terms and the last term of the above expression will be left unchanged.  To estimate the remaining terms, we will use the elementary inequalities $(a+b)^2\leq 2( a^2+b^2)$ and
\begin{align}
&K^{q - 1} \leq \left({M \over \delta}\right)^{q-1} + {\delta \over M}K^q, \\
&K^{q - 2} \leq \left({M \over \delta}\right)^{q-2} + \left({\delta \over M}\right)^2K^q.
\end{align}
The inequalites involving $K$ are obtained by looking at the cases $K\leq M/\delta$ and $K>M/\delta$. The first  holds for every $q \geq 1$ and the second for every $q \geq 2$.  Note that for $q = 1$, the term containing $K^\epsilon(s,x_s^\epsilon)^{q - 2}$ vanishes. Applying these inequalities yields 
\begin{align}\label{K_bound1}
&K^\epsilon(t,x_t^\epsilon)^q\\
\leq& e^{-\alpha t/\epsilon}K^\epsilon(0,x_0^\epsilon)^q +\frac{D}{\alpha}-\frac{d}{\epsilon}e^{-\alpha t/\epsilon}\int_0^te^{\alpha s/\epsilon}K^\epsilon(s,x_s^\epsilon)^q ds\notag\\
&+\frac{q}{\sqrt{\epsilon}}e^{-\alpha t/\epsilon}\int_0^te^{\alpha s/\epsilon}K^\epsilon(s,x_s^\epsilon)^{q-1}(\nabla_{z} K)^\epsilon(s,x_s^\epsilon)\cdot\sigma(s,x^\epsilon_s) dW_s,\notag
\end{align}
where
\begin{align}
D=&qM\left({M \over \delta}\right)^{q-1}\left[{\lambda }+\epsilon+{\sqrt{\epsilon}}\|-\partial_t\psi-\nabla_q V+F\|_\infty+\frac{1}{2}\|\Sigma\|_{F,\infty}\right]\\
&+{q(q-1)M^2}(M/\delta)^{q-2}\|\Sigma\|_\infty,\notag\\
d=&{qc\lambda}-\alpha-qC\epsilon-{q\delta}{\sqrt{\epsilon}}\|-\partial_t\psi-\nabla_q V+F\|_\infty-\frac{q\delta}{2}\|\Sigma\|_{F,\infty}-{q(q-1)\delta^2}\|\Sigma\|_\infty\notag\\
&-{q(q-1)\delta^2}\|\Sigma\|_\infty-\left({q\lambda }+q\epsilon+{q}{\sqrt{\epsilon}}\|-\partial_t\psi-\nabla_q V+F\|_\infty+\frac{q}{2}\|\Sigma\|_{F,\infty}\right){\delta}.\notag
\end{align}

For all $\epsilon$, $\delta$, $\alpha$ sufficiently small, $d$ is non-negative, and hence
\begin{align}
&K^\epsilon(t,x_t^\epsilon)^q\leq K^\epsilon(0,x_0^\epsilon)^q +\frac{D}{\alpha}\\
&+\frac{q}{\sqrt{\epsilon}}e^{-\alpha t/\epsilon}\int_0^te^{\alpha s/\epsilon}K^\epsilon(s,x_s^\epsilon)^{q-1}(\nabla_{z} K)^\epsilon(s,x_s^\epsilon)\cdot\sigma(s,x^\epsilon_s) dW_s.\notag
\end{align}
By Assumption  \ref{assump2}, $K^\epsilon(0,x_0^\epsilon)$ is bounded, so we are done.
\end{proof}
We will use this bound to prove several results about the behavior of the kinetic energy and momentum as $\epsilon\rightarrow 0^+$.

\subsection{Integrability of  the Kinetic Energy}
\begin{lemma}\label{E_K_finite}
Under Assumptions \ref{assump1} and \ref{assump2},  $E[\sup_{t\in[0,T]}K^\epsilon(t,x_t^\epsilon)^{p}]$ is finite for any $T>0$, $\epsilon>0$, and  $p>0$.
\end{lemma}
\begin{proof}
Fix $T>0$, $\epsilon>0$. First, let $p>2$.  Given $M>0$, define the stopping time $\tau_M=\inf\{t:K^\epsilon(t,x^\epsilon_t)=M\}$. By Assumption \ref{assump2}  we can take $M$ large enough so that $K^\epsilon(0,x^\epsilon_0)<M$.

Let $t\leq T$  and raise \req{K_bound1} (with $q=1$) to the $p$th power to obtain
\begin{align}
& K^\epsilon(t\wedge\tau_M,(x^\epsilon)_t^{\tau_M})^p\notag\\
\leq & C_1+C_2 \left(\int_0^t 1_{s\leq \tau_M}K^\epsilon(s\wedge \tau_M,(x^\epsilon)_s^{\tau_M})ds\right)^p\\
&+C_3\bigg|\int_0^t1_{s\leq \tau_M} e^{\alpha (s\wedge\tau_M)/\epsilon}( \nabla_z K)^\epsilon(s\wedge \tau_M,(x^\epsilon)_s^{\tau_M})\cdot \sigma(s\wedge \tau_M,(x^\epsilon)^{\tau_M}_s) dW_s\bigg|^p\notag
\end{align}
where $C_i$ are constants (that depend on $\epsilon$ and $T$).

Therefore, applying H\"older's inequality to the second term and the  Burkholder-Davis-Gundy inequality  to the third term, (see, for example, Theorem 3.28 in \cite{karatzas2014brownian}), we obtain
\begin{align}
&E\left[\sup_{s\in [0,t]}K^\epsilon(s\wedge\tau_M,(x^\epsilon)_s^{\tau_M})^p\right]\notag\\
\leq &C_1+C_2 T^{p-1} \int_0^t E[K^\epsilon(s\wedge \tau_M,(x^\epsilon)_s^{\tau_M})^p]ds\\
&+C_4E\bigg[\bigg(\int_0^t1_{r\leq \tau_M} e^{2\alpha (r\wedge\tau_M)/\epsilon}\|(\nabla_z K)^\epsilon(r\wedge \tau_M,(x^\epsilon)_r^{\tau_M})\|^2\notag\\
&\hspace{20mm}\times\| \sigma(r\wedge \tau_M,(x^\epsilon)^{\tau_M}_r)\|^2 dr\bigg)^{p/2}\bigg]\notag\\
\leq & C_1+C_2 T^{p-1} \int_0^t E\left[\sup_{r\in[0,s]}K^\epsilon(r\wedge \tau_M,(x^\epsilon)_r^{\tau_M})^p\right]ds\\
&+C_5E\left[\left(\int_0^t(1+ K^\epsilon(r\wedge \tau_M,(x^\epsilon)_r^{\tau_M}))^2  dr\right)^{p/2}\right].\notag
\end{align}

By assumption, $p>2$, so we can use H\"older's inequality again to obtain
\begin{align}
&E\left[\sup_{s\in [0,t]}K(s\wedge\tau_M,(x^\epsilon)_s^{\tau_M})^p\right]\notag\\
\leq &C_1+C_2 T^{p-1} \int_0^t E\left[\sup_{r\in[0,s]}K^\epsilon(r\wedge \tau_M,(x^\epsilon)_r^{\tau_M})^p\right]ds\\
&+C_5T^{p/2-1}E\left[\int_0^t(1+ K^\epsilon(r\wedge \tau_M,(x^\epsilon)_r^{\tau_M}))^p dr\right]\notag\\
\leq &C_6+C_7 \int_0^t E\left[\sup_{r\in[0,s]}K^\epsilon(r\wedge \tau_M,(x^\epsilon)_r^{\tau_M})^p\right]ds,\label{Gronwall_rhs}
\end{align}
where $C_i$ are independent of $t$ and $M$.

By the definition of $\tau_M$,
\begin{equation}
\sup_{s\in [0,t]}K^\epsilon(s\wedge\tau_M,(x^\epsilon)_s^{\tau_M})^p\leq M
\end{equation}
for all $t$.  Therefore the integral in \req{Gronwall_rhs} is finite for all $t\leq T$ and Gronwall's inequality gives
\begin{align}
E\left[\sup_{t\in [0,T]}K^\epsilon(t\wedge\tau_M,(x^\epsilon)_t^{\tau_M})^p\right]\leq C_6e^{C_7T}.
\end{align}
The $C_i$ are independent of $M$, so taking $M\rightarrow\infty$ and using the Monotone Convergence Theorem implies
\begin{align}
E\left[\sup_{t\in [0,T]}K^\epsilon(t,x^\epsilon_t)^p\right]\leq C_6e^{C_7T}<\infty.
\end{align}
This gives the result for  $p>2$.  It follows for all $p>0$ by an application of H\"older's inequality.
\end{proof}

\subsection{Supremum of the Expectation of the Kinetic Energy}
Combining Lemmas \ref{K_bound_lemma} and \ref{E_K_finite} we can prove the following bound for the supremum of the expected value of the kinetic energy.
\begin{proposition}\label{Sup_E_prop}
Under the Assumptions \ref{assump1} and \ref{assump2}, for any $T>0$, $q>0$ we have
\begin{align}
\sup_{t\in[0,T]}E[ K^\epsilon(t,x^\epsilon_t)^{q}]=O(1) \text{ as $\epsilon\rightarrow 0^+$.}
\end{align}
\end{proposition}
\begin{proof}
First take $T>0$, $q\in\mathbb{N}$, $q\geq 1$.  The following computation shows that
\begin{equation}
M_t\equiv\int_0^t e^{\alpha s/\epsilon}K^\epsilon(s,x_s^\epsilon)^{q-1}(\nabla_{z} K)^\epsilon(s,x_s^\epsilon)\cdot\sigma(s,x^\epsilon_s) dW_s
\end{equation}
 is a martingale (see \cite{karatzas2014brownian}):
\begin{align}\label{martingale_proof}
&E\left[\int_0^t\|e^{\alpha s/\epsilon}K^\epsilon(s,x_s^\epsilon)^{q-1}(\nabla_{z} K)^\epsilon(s,x_s^\epsilon)\cdot\sigma(s,x^\epsilon_s)\|^2 ds\right]\\
\leq& e^{2\alpha t/\epsilon}\|\sigma\|^2_{\infty} t E\left[\sup_{s\in[0,t]} K^\epsilon(s,x_s^\epsilon)^{2(q-1)}(M+K^\epsilon(s,x_s^\epsilon))^2\right]\notag\\
\leq & 2e^{2\alpha t/\epsilon}\|\sigma\|^2_{\infty} t \left(M^2E\left[\sup_{s\in[0,t]}K^\epsilon(s,x_s^\epsilon)^{2(q-1)}\right]+E\left[\sup_{s\in[0,t]}K^\epsilon(s,x_s^\epsilon)^{2q}\right]\right)<\infty,\notag
\end{align}
where we used Assumption \ref{assump1} and Lemma \ref{E_K_finite}.\\

Therefore,  taking the expectation of \req{K_bound_lemma_eq}, we see that there exists $\kappa>0$ such that for all $t\leq T$ and all $\alpha$ and $\epsilon$ sufficiently small, we have
\begin{align}
E[ K^\epsilon(t,x_t^\epsilon)^q]\leq \frac{\kappa}{\alpha}.
\end{align}
 This proves the result for $q$ a positive integer. The result then follows for arbitrary $q>0$ by an application of H\"older's inequality.

\end{proof}
\begin{corollary}
We note that if the constants involved in the bounds from Assumptions \ref{assump1}-\ref{assump2} are valid uniformly for $t\in[0,\infty)$ (and not just $t\in[0,T]$) then we obtain the stronger bound
\begin{align}
\sup_{t\in[0,\infty)}E[ K^\epsilon(t,x^\epsilon_t)^{q}]=O(1) \text{ as $\epsilon\rightarrow 0^+$}
\end{align}
for any $q>0$.
\end{corollary}

\subsection{Expectation of the Supremum of the Kinetic Energy}
We now have the ingredients to derive a bound on the expectation of the supremum of the kinetic energy.  For this, we need to recall a special case of  Lemma 5.1 from \cite{particle_manifold_paper}:
\begin{lemma}\label{matrix_exp_decay_bound}
 Let $V\in L^2_{loc}(dt)$ be an $\mathbb{R}^{ k}$-valued process. For any $\alpha>0$, $T\geq \delta>0$  we have the $P$-a.s. bound  
\begin{align}\label{Phi_int_bound2}
&\sup_{t\in[0,T]}\left|\int_0^t e^{-\alpha (t-s)}V_s dW_s\right|\\
\leq &5\left(e^{-\alpha\delta}\sup_{t\in[0,T]}\left|\int_0^t V_r dW_r\right|+\max_{k=0,...,N-1}\sup_{t\in [k\delta, (k+2)\delta]}\left|\int_{k\delta}^tV_rdW_r\right|\right)\notag
\end{align}
where $N=\max\{k\in\mathbb{Z}:k\delta<T\}$.
\end{lemma}

\begin{proposition}\label{E_sup_prop}
Under Assumptions \ref{assump1} and \ref{assump2}, for any $T>0$, $p>0$, $\beta>0$ we have
\begin{align}
E\left[\sup_{t\in[0,T]}K^\epsilon(t,x_t^\epsilon)^{p}\right]=O(\epsilon^{-\beta}) \text{ as $\epsilon\rightarrow 0^+$.}
\end{align}
\end{proposition}
\begin{proof}

By Lemma \ref{K_bound_lemma} with $q=1$, there exist $\alpha>0$ and $\kappa>0$ such that for all $\epsilon$ sufficiently small and all $t\in[ 0,T]$, the following bound holds a.s.:
\begin{align}
 K^\epsilon(t,x_t^\epsilon)\leq \frac{\kappa}{\alpha}+\frac{1}{\sqrt{\epsilon}} e^{-\alpha t/\epsilon}\int_0^t e^{\alpha s/\epsilon}(\nabla_{z} K)^\epsilon(s,x_s^\epsilon)\cdot\sigma(s,x^\epsilon_s) dW_s.
\end{align}
We will first prove the proposition under the additional assumption $p>2$.  The general case $p >0$ will follow by an application of H\"older's inequality.
\begin{align}
& E\left[\sup_{t\in[0,T]}K^\epsilon(t,x_t^\epsilon)^p\right]\\
\leq& 2^{p-1}(\kappa/\alpha)^p+ \frac{2^{p-1}}{\epsilon^{p/2}}E\left[\sup_{t\in[0,T]}\left|\int_0^t e^{-\alpha (t-s)/\epsilon}(\nabla_{z} K)^\epsilon(s,x_s^\epsilon)\cdot\sigma(s,x^\epsilon_s) dW_s\right|^p\right].\notag
\end{align}
For any $T\geq\delta>0$, Lemma \ref{matrix_exp_decay_bound} (with $\alpha/\epsilon$ in place of $\alpha$) implies
\begin{align}\label{E_sup_bound1}
&E\left[\sup_{t\in[0,T]}\left|\int_0^t e^{-\alpha (t-s)/\epsilon}(\nabla_{z} K)^\epsilon(s,x_s^\epsilon)\cdot\sigma(s,x^\epsilon_s) dW_s\right|^p\right]\\
\leq& 5^p2^{p-1}\left(e^{-p\alpha\delta/\epsilon}E\left[\sup_{t\in[0,T]}\left|\int_0^t(\nabla_{z} K)^\epsilon(s,x_s^\epsilon)\cdot\sigma(s,x^\epsilon_s) dW_s\right|^p\right]\right.\notag\\
&\left.+E\left[\max_{k=0,...,N-1}\sup_{t\in [k\delta, (k+2)\delta]}\left|\int_{k\delta}^t (\nabla_{z} K)^\epsilon(s,x_s^\epsilon)\cdot\sigma(s,x^\epsilon_s) dW_s\right|^p\right]\right)\notag
\end{align}
where $N=\max\{k\in\mathbb{Z}:k\delta<T\}$. 

The Burkholder-Davis-Gundy inequality, applied to the first term on the right side of the inequality, implies existence of a constant $\tilde C>0$ such that
\begin{align}
&E\left[\sup_{t\in[0,T]}\left|\int_0^t (\nabla_{z} K)^\epsilon(s,x_s^\epsilon)\cdot\sigma(s,x^\epsilon_s) dW_r\right|^p\right]\\
\leq &\tilde C E\left[\left(\int_0^T  \|(\nabla_{z} K)^\epsilon(s,x_s^\epsilon)\cdot\sigma(s,x^\epsilon_s) \|^2dr\right)^{p/2}\right]\notag\\
\leq & \tilde C\|\sigma\|_\infty^p E\left[\left(\int_0^T  \|(\nabla_{z} K)^\epsilon(s,x_s^\epsilon) \|^2dr\right)^{p/2}\right]\notag\\
\leq & \tilde C\|\sigma\|_\infty^p E\left[\left(\int_0^T  (M+ K^\epsilon(s,x_s^\epsilon) )^2dr\right)^{p/2}\right].\notag
\end{align}
In the last line, we used Assumption \ref{assump1}.

We have assumed $p> 2$, so we can use H\"older's inequality with exponents $p/(p-2)$ and $p/2$,  to get
\begin{align}
&E\left[\sup_{t\in[0,T]}\left|\int_0^t (\nabla_{z} K)^\epsilon(s,x_s^\epsilon)\cdot\sigma(s,x^\epsilon_s) dW_r\right|^p\right]\\
\leq & \tilde C\|\sigma\|_\infty^p T^{p/2-1} E\left[\int_0^T  (M+ K^\epsilon(s,x_s^\epsilon) )^pdr\right]\notag\\
\leq & 2^{p-1}\tilde C\|\sigma\|_\infty^p T^{p/2}\left(  M^p+\sup_{s\in[0,T]} E[K^\epsilon(s,x_s^\epsilon) ^p]\right)\notag\\
=&O(1)\notag
\end{align}
as $\epsilon\rightarrow 0^+$ by Proposition \ref{Sup_E_prop}.

We now work on the second term in \req{E_sup_bound1}. Using the fact that the $\ell^\infty$-norm on $\mathbb{R}^N$ is bounded by the $\ell^{\tilde p}$ norm for any $\tilde p\geq 1$, and then applying H\"older's inequality and the  Burkholder-Davis-Gundy inequality, we derive the bound
\begin{align}
&E\left[\max_{k=0,...,N-1}\sup_{t\in [k\delta, (k+2)\delta]}\left|\int_{k\delta}^t (\nabla_{z} K)^\epsilon(s,x_s^\epsilon)\cdot\sigma(s,x^\epsilon_s) dW_r\right|^p\right]\\
\leq &E\left[\left(\sum_{k=0}^{N-1}\sup_{t\in [k\delta, (k+2)\delta]}\left|\int_{k\delta}^t(\nabla_{z} K)^\epsilon(s,x_s^\epsilon)\cdot\sigma(s,x^\epsilon_s)dW_r\right|^{p\tilde p}\right)^{1/\tilde p}\right]\notag\\
\leq &\left(\sum_{k=0}^{N-1}E\left[\sup_{t\in [k\delta, (k+2)\delta]}\left|\int_{k\delta}^t (\nabla_{z} K)^\epsilon(s,x_s^\epsilon)\cdot\sigma(s,x^\epsilon_s)dW_r\right|^{p\tilde p}\right]\right)^{1/\tilde p}\notag\\
\leq &\left(\sum_{k=0}^{N-1}\tilde C E\left[\left(\int_{k\delta}^{(k+2)\delta}\| (\nabla_{z} K)^\epsilon(s,x_s^\epsilon)\cdot\sigma(s,x^\epsilon_s)\|^2dr\right)^{p\tilde p/2}\right]\right)^{1/\tilde p}\notag\\
\leq &\tilde C^{1/\tilde p}\|\sigma\|_\infty^p\left(\sum_{k=0}^{N-1} E\left[\left(\int_{k\delta}^{(k+2)\delta}\|(\nabla_{z} K)^\epsilon(s,x_s^\epsilon)\|^2dr\right)^{p\tilde p/2}\right]\right)^{1/\tilde p}.\notag
\end{align}
Note that for $0\leq k<N$ we have $0\leq (k+2)\delta\leq (N+1)\delta\leq 2T$. So here, the time interval corresponding to $\|\cdot\|_\infty$ can be taken to be $[0,2T]$.

By assumption, $p\tilde p> 2$, so using H\"older's inequality again with exponents  $p\tilde p/(p\tilde p-2)$ and $p\tilde p/2$, along with Assumption \ref{assump1}, we get
\begin{align}
&E\left[\max_{k=0,...,N-1}\sup_{t\in [k\delta, (k+2)\delta]}\left|\int_{k\delta}^t (\nabla_{z} K)^\epsilon(s,x_s^\epsilon)\cdot\sigma(s,x^\epsilon_s) dW_r\right|^p\right]\\
\leq &\tilde C^{1/\tilde p}\|\sigma\|_\infty^p\left(\sum_{k=0}^{N-1}(2\delta)^{p\tilde p/2-1}\int_{k\delta}^{(k+2)\delta}E[\|(\nabla_{z} K)^\epsilon(s,x_s^\epsilon)\|^{p\tilde p}]dr\right)^{1/\tilde p}\notag\\
\leq &\tilde C^{1/\tilde p}\|\sigma\|_\infty^p\left((2\delta)^{p\tilde p/2}N\right)^{1/\tilde p}\sup_{s\in[0,(N+1)\delta]}E[(M+ K^\epsilon(s,x_s^\epsilon))^{p\tilde p}]^{1/\tilde p}.\notag
\end{align}
Using $N<T/\delta$ we obtain
\begin{align}
&E\left[\max_{k=0,...,N-1}\sup_{t\in [k\delta, (k+2)\delta]}\left|\int_{k\delta}^t (\nabla_{z} K)^\epsilon(s,x_s^\epsilon)\cdot\sigma(s,x^\epsilon_s) dW_r\right|^p\right]\\
\leq &2^{p/2}\tilde C^{1/\tilde p}T^{1/\tilde p}\|\sigma\|_\infty^p\delta^{p/2-1/\tilde p}\sup_{s\in[0,2T]}E[(M+K^\epsilon(s,x_s^\epsilon))^{p\tilde p}]^{1/\tilde p}\notag\\
=&\delta^{p/2-1/\tilde p} O(1),\notag
\end{align}
where we used Proposition \ref{Sup_E_prop}.

Combining these results we see that for  $\epsilon>0$ sufficiently small and any $T\geq \delta>0$, $\tilde p\geq 1$ we have
\begin{align}\label{E_sup_bound}
&E\left[\sup_{t\in[0,T]}\left|\int_0^t e^{-\alpha (t-s)/\epsilon}(\nabla_{z} K)^\epsilon(s,x_s^\epsilon)\cdot\sigma(s,x^\epsilon_s) dW_s\right|^p\right]\\
\leq & e^{-p\alpha\delta/\epsilon}O(1)+\delta^{p/2-1/\tilde p} O(1),\notag
\end{align}
where the big-O terms do not depend on $\delta$.

Now let $0<\xi<1$ and choose $\delta=\epsilon^{1-\xi}$.  Then

\begin{align}\label{E_sup_bound2}
& E\left[\sup_{t\in[0,T]}K^\epsilon(t,x_t^\epsilon)^p\right]\leq 2^{p-1}(\kappa/\alpha)^p+ \frac{2^{p-1}}{\epsilon^{p/2}}\left( e^{-p\alpha/\epsilon^{\xi}}O(1)+\epsilon^{(1-\xi)(p/2-1/\tilde p)} O(1)\right)\notag\\
=&2^{p-1}\left((\kappa/\alpha)^p+ \epsilon^{-p/2} e^{-p\alpha/\epsilon^{\xi}}O(1)+\epsilon^{(1-\xi)(p/2-1/\tilde p)-p/2} O(1)\right).
\end{align}
For any $\beta>0$  there exists $\tilde p\geq 1$ and $0<\xi<1$ such that 
\begin{align}
(1-\xi)(p/2-1/\tilde p)-p/2= -\xi{p \over 2} - {1 - \xi \over \tilde{p}} > -\beta.
\end{align}
Hence the term $\epsilon^{(1-\xi)(p/2-1/\tilde p)-p/2} O(1)$ diverges more slowly than $\epsilon^{-\beta}$.
 Also,  $\epsilon^{-p/2} e^{-p\alpha/\epsilon^\xi}=o(1)$ for all $\xi>0$. This proves the result for $p>2$.  The result for all $p>0$ again follows by an application of H\"older's inequality.

\end{proof}

\subsection{Decay of Momentum}
Starting in this section, we will assume that the difference between the canonical momentum and $\psi$ is bounded by the kinetic energy in the following sense:
\begin{assumption}\label{assump3}
We assume that for every $T>0$ there exists $c>0$, $\eta>0$ such that
\begin{align}
K(\epsilon,t,q,z)\geq c\|z\|^{2\eta}
\end{align}
on $(0,\epsilon_0]\times [0,T]\times\mathbb{R}^{2n}$.
\end{assumption}

With the addition of Assumption \ref{assump3}, the bounds on the kinetic energy from Propositions \ref{Sup_E_prop} and \ref{E_sup_prop} yield the following decay rates for the momentum to the submanifold defined by $p=\psi(t,q)$:
\begin{lemma}\label{p_decay_lemma}
Under Assumptions \ref{assump1}-\ref{assump3}, for any $T>0$, $p>0$ we have
\begin{align}
\sup_{t\in[0,T]}E[\|p_t^\epsilon-\psi(t,q_t^\epsilon)\|^{p}]=O(\epsilon^{p/2}) \text{ as }\epsilon\rightarrow 0^+
\end{align}
and for any $p>0$,  $T>0$, $0<\beta<p/2$ we have
\begin{align}
E\left[\sup_{t\in[0,T]}\|p_t^\epsilon-\psi(t,q_t^\epsilon)\|^{p}\right]=O(\epsilon^\beta) \text{ as }\epsilon\rightarrow 0^+.
\end{align}
\end{lemma}

Our final momentum decay rate result concerns  a class of integrals with respect to products of the components of $u_t^\epsilon\equiv p_t^\epsilon-\psi(t,q_t^\epsilon)$.
\begin{proposition}\label{p_int_decay}
Let $f:[0,\infty)\times\mathbb{R}^n\rightarrow\mathbb{R}$ be a $C^1$ function, such that for every $T>0$, $f$, $\partial_t f$, and $\nabla_q f$ are bounded on $[0,T]\times\mathbb{R}^n$. Define $u_t^\epsilon=p_t^\epsilon-\psi(t,q_t^\epsilon)$. Under Assumptions \ref{assump1}-\ref{assump3},  for any $p>0$, $T>0$, $i,j=1,...,n$ we have
\begin{align}
E\left[\sup_{t\in[0,T]}\left|\int_0^t f(s,q^\epsilon_s)d ((u^\epsilon_s)_i(u^\epsilon_s)_j)\right|^p\right]=O(\epsilon^{p/2})\text{ as }\epsilon\rightarrow 0^+.
\end{align} 
\end{proposition}
\begin{proof}
$f(s,q^\epsilon_s)$ is a $C^1$-semimartingle.  Therefore integration by parts gives
\begin{align}
&\int_0^t f(s,q^\epsilon_s)d ((u^\epsilon_s)_i(u^\epsilon_s)_j)=f(t,q^\epsilon_t)(u^\epsilon_t)_i(u^\epsilon_t)_j-f(0,q^\epsilon_0)(u^\epsilon_0)_i(u^\epsilon_0)_j\\
&-\int_0^t(u^\epsilon_s)_i(u^\epsilon_s)_j(\partial_s f(s,q^\epsilon_s)+\nabla_q f(s,q^\epsilon_s)\cdot\nabla_p H^\epsilon(s,x_s^\epsilon))ds.\notag
\end{align}
Hence, for $p\geq 1$, using Assumption 3 we obtain:
\begin{align}
&E\left[\sup_{t\in[0,T]}\left|\int_0^t f(s,q^\epsilon_s)d ((u^\epsilon_s)_i(u^\epsilon_s)_j)\right|^p\right]\\
\leq& 3^{p-1}\bigg(E\left[\sup_{t\in[0,T]}|f(t,q^\epsilon_t)(u^\epsilon_t)_i(u^\epsilon_t)_j|^p\right]+E[|f(0,q^\epsilon_0)(u^\epsilon_0)_i(u^\epsilon_0)_j|^p]\notag\\
&+E\left[\sup_{t\in[0,T]}\left|\int_0^t(u^\epsilon_s)_i(u^\epsilon_s)_j(\partial_s f(s,q^\epsilon_s)+\nabla_q f(s,q^\epsilon_s)\cdot\nabla_p K^\epsilon(s,x_s^\epsilon))ds\right|^p\right]\bigg)\notag\\
\leq &3^{p-1}\bigg(2\|f\|_\infty^pE\left[\sup_{t\in[0,T]}\epsilon^p(K^\epsilon(t,x_t^\epsilon)/c)^{p/\eta}\right]\notag\\
&+E\left[\left(\int_0^T\|u_s^\epsilon\|^2(\|\partial_s f\|_\infty+\|\nabla_q f\|_\infty\|(\nabla_z K)^\epsilon(s,x_s^\epsilon)/\sqrt{\epsilon}\|)ds\right)^p\right]\bigg).\notag
\end{align}
Now, using Assumption \ref{assump1} and Proposition \ref{E_sup_prop}, for any $\beta>0$ we find
\begin{align}
&E\left[\sup_{t\in[0,T]}\left|\int_0^t f(s,q^\epsilon_s)d ((u^\epsilon_s)^i(u^\epsilon_s)^j)\right|^p\right] \\
\leq &O(\epsilon^{p-\beta}) +3^{p-1}E\bigg[\bigg(\int_0^T\epsilon (K^\epsilon(s,x_s^\epsilon)/c)^{1/\eta}\notag\\
&\hspace{30mm}\times(\|\partial_s f\|_\infty+\epsilon^{-1/2}\|\nabla_q f\|_\infty(M+K^\epsilon(s,x_s^\epsilon)))ds\bigg)^p\bigg].\notag
\end{align}
 H\"older's inequality and Proposition \ref{Sup_E_prop} allow us to bound the second term:
\begin{align}
&E\left[\sup_{t\in[0,T]}\left|\int_0^t f(s,q^\epsilon_s)d ((u^\epsilon_s)_i(u^\epsilon_s)_j)\right|^p\right] \\
\leq &O(\epsilon^{p-\beta}) +3^{p-1}c^{-p/\eta}T^{p-1}\epsilon^p E\bigg[\int_0^T K^\epsilon(s,x_s^\epsilon)^{p/\eta}(\|\partial_s f\|_\infty\notag\\
&\hspace{42mm}+\epsilon^{-1/2}\|\nabla_q f\|_\infty(M+K^\epsilon(s,x_s^\epsilon)))^pds \bigg]\notag\\
\leq &O(\epsilon^{p-\beta}) +3^{2(p-1)}c^{-p/\eta}T^{p}\epsilon^p\sup_{t\in[0,T]}E[ K^\epsilon(s,x_s^\epsilon)^{p/\eta}(\|\partial_s f\|_\infty^p\notag\\
&\hspace{42mm}+\epsilon^{-p/2}\|\nabla_q f\|_\infty^p(M^p+K^\epsilon(s,x_s^\epsilon)^p))]\notag\\
=&O(\epsilon^{p-\beta}) +3^{2(p-1)}c^{-p/\eta}T^{p}\epsilon^pO(\epsilon^{-p/2}).\notag
\end{align}
Taking $\beta=p/2$ gives the result when $p\geq1$.  The result for any  $p>0$ follows from H\"older's inequality.

\end{proof}

\section{Derivation of the Limiting Equation}\label{sec:limit_eq}
In this section, we derive the equation satisfied by  $q^\epsilon_t$ in the limit $\epsilon\rightarrow 0^+$.  The actual convergence proof will be given in the following section. The derivation is an adaptation of the methods used in  \cite{Hottovy2014,particle_manifold_paper}.  We will need the following:
\begin{assumption}\label{assump4}
We assume that $\gamma$ is $C^1$ and is independent of $p$.
\end{assumption}

The starting point for the derivation is a rewriting of Hamilton's equation of motion in terms of the variables $u_t^\epsilon\equiv p_t^\epsilon-\psi(t,q_t^\epsilon)$:
\begin{align}
d(u^\epsilon_t)_i=&-\gamma_{ij}(t,q_t^\epsilon)\partial_{p_j}H^\epsilon(t,x^\epsilon_t)dt+(-\partial_{q^i}H^\epsilon(t,x_t^\epsilon)+F_i(t,x^\epsilon_t))dt\\
&-\partial_t\psi_i(t,q_t^\epsilon)dt-\partial_{q^k}\psi_i(t,q^\epsilon_t)\partial_{p_k}H^\epsilon(t,x_t^\epsilon)dt+\sigma_{ij}(t,x_t^\epsilon)dW^j_t\notag\\
=&-\tilde\gamma_{ik}(t,q_t^\epsilon)\partial_{p_k}K^\epsilon(t,x^\epsilon_t)dt-(\partial_{q^i}K)^\epsilon(t,x_t^\epsilon)dt\label{z_eq}\\
&+(-\partial_t\psi_i(t,q_t^\epsilon)-\partial_{q^i}V(t,q_t^\epsilon)+F_i(t,x^\epsilon_t))dt+\sigma_{i\rho}(t,x_t^\epsilon)dW^\rho_t\notag
\end{align}
where
\begin{align}\label{tilde_gamma_def}
\tilde\gamma_{ik}(t,q)\equiv\gamma_{ik}(t,q) +\partial_{q^k}\psi_i(t,q)-\partial_{q^i}\psi_k(t,q).
\end{align}
The second and third terms in $\tilde \gamma$ together form an antisymmetric matrix, hence the eigenvalue bound for $\gamma$ from Assumption \ref{assump2} implies invertibility of $\tilde\gamma$. See Lemma \ref{eig_bound_lemma1}. We define the components of $\tilde\gamma^{-1}$ such that 
\begin{align}\label{tilde_gamma_inv_def}
(\tilde\gamma^{-1})^{ij}\tilde\gamma_{jk}=\delta^i_k,
\end{align}
and for any $v\in\mathbb{R}^n$ we define $(\tilde\gamma^{-1}v)^i=(\tilde\gamma^{-1})^{ij}v_j$.

This lets us solve for $\nabla_p H^\epsilon(t,x_t^\epsilon)dt$ to get
\begin{align}
d(q_t^\epsilon)^i=&\partial_{p_i} H^\epsilon(t,x^\epsilon_t)dt \\
=&(\tilde\gamma^{-1})^{ij}(t,q_t^\epsilon)(-\partial_t\psi_j(t,q_t^\epsilon)-\partial_{q^j}V(t,q_t^\epsilon)+F_j(t,x^\epsilon_t))dt\notag\\
&-(\tilde\gamma^{-1})^{ij}(t,q_t^\epsilon)(\partial_{q^j}K)^\epsilon(t,x_t^\epsilon)dt+(\tilde\gamma^{-1})^{ij}(t,q_t^\epsilon)\sigma_{j\rho}(t,x_t^\epsilon)dW^\rho_t\notag\\
&-(\tilde\gamma^{-1})^{ij}(t,q_t^\epsilon)d(u^\epsilon_t)_j.\notag
\end{align}
$\tilde\gamma^{-1}(t,q^\epsilon_t)$ is pathwise $C^1$, so integrating the last term by parts results in
\begin{align}
&-(\tilde\gamma^{-1})^{ij}(t,q_t^\epsilon)d(u^\epsilon_t)_j=- d((\tilde\gamma^{-1})^{ij}(t,q_t^\epsilon)(u^\epsilon_t)_j)+(u_t^\epsilon)_j\partial_t(\tilde\gamma^{-1})^{ij}(t,q_t^\epsilon)dt\\
&+(u_t^\epsilon)_j\partial_{q^l}(\tilde\gamma^{-1})^{ij}(t,q_t^\epsilon)\partial_{p_l}H^\epsilon(t,x_t^\epsilon)dt.\notag
\end{align}
Therefore
\begin{align}\label{q_eq1}
d(q_t^\epsilon)^i=&(\tilde\gamma^{-1})^{ij}(t,q_t^\epsilon)(-\partial_t\psi_j(t,q_t^\epsilon)-\partial_{q^j}V(t,q_t^\epsilon)+F_j(t,x^\epsilon_t))dt\\
&-(\tilde\gamma^{-1})^{ij}(t,q_t^\epsilon)(\partial_{q^j}K)^\epsilon(t,x_t^\epsilon)dt+(\tilde\gamma^{-1})^{ij}(t,q_t^\epsilon)\sigma_{j\rho}(t,x_t^\epsilon)dW^\rho_t\notag\\
&- d((\tilde\gamma^{-1})^{ij}(t,q_t^\epsilon)(u^\epsilon_t)_j)+(u_t^\epsilon)_j\partial_t(\tilde\gamma^{-1})^{ij}(t,q_t^\epsilon)dt\notag\\
&+(u_t^\epsilon)_j\partial_{q^l}(\tilde\gamma^{-1})^{ij}(t,q_t^\epsilon)\partial_{p_l}H^\epsilon(t,x_t^\epsilon)dt.\notag
\end{align}

In order to  homogenize $(u_t^\epsilon)_j\partial_{p_l}H^\epsilon(t,x_t^\epsilon)dt$, we make  the additional assumption:
\begin{assumption}\label{assump5}
We assume that $K$ has the form
\begin{align}
K(\epsilon,t,q,z)=\tilde K(\epsilon,t,q,A^{ij}(t,q)z_iz_j)
\end{align}
where $\tilde K(\epsilon,t,q,\zeta)$ is $C^2$ in $(t,q,\zeta)$ for every $\epsilon$, non-negative on $(0,\epsilon_0]\times [0,\infty)\times\mathbb{R}^n\times[0,\infty)$, $A(t,q)$ is a $C^2$ function whose values are symmetric $n \times n$-matrices.   We also assume that for every $T>0$, the eigenvalues of $A$ are bounded above and below by some constants $C>0$ and $c>0$ respectively, uniformly on $[0,T]\times\mathbb{R}^n$.  

We will write $\tilde K^\prime$ for $\partial_\zeta \tilde K$ and will use the abbreviation $\|z\|_A^2$ for $A^{ij}(t,q)z_iz_j$ when the implied values of $t$ and $q$ are apparent from the context.
\end{assumption}

With this assumption, 
\begin{align}\label{dq_K_eq}
(\partial_{q^i}K)^\epsilon(t,x_t^\epsilon)=&\partial_{q^i}\tilde K(\epsilon,t,q_t^\epsilon,\|u_t^\epsilon\|^2_A/\epsilon)\\
&+\tilde K^\prime(\epsilon,t,q_t^\epsilon,\|u_t^\epsilon\|^2_A/\epsilon) \partial_{q^i}A^{kl}(t,q_t^\epsilon)(u_t^\epsilon)_k(u_t^\epsilon)_l/\epsilon\notag
\end{align}
and
\begin{align}\label{z_dH_eq}
&\partial_{p_l}H^\epsilon(t,x_t^\epsilon)=\frac{2}{\epsilon}A^{lk}(t,q_t^\epsilon) \tilde K^\prime(\epsilon,t,q_t^\epsilon,\|u_t^\epsilon\|_A^2/\epsilon)(u_t^\epsilon)_k.
\end{align}
To simplify $(u_t^\epsilon)_j\partial_{p_l}H^\epsilon(t,x_t^\epsilon)dt$, we compute
\begin{align}\label{Scott's_trick}
&d((u_t^\epsilon)_i(u_t^\epsilon)_j)=(u_t^\epsilon)_id(u_t^\epsilon)_j+(u_t^\epsilon)_jd(u_t^\epsilon)_i+d[u^\epsilon_i,u^\epsilon_j]_t\\
=&(-(u_t^\epsilon)_i\tilde\gamma_{jk}(t,q_t^\epsilon)-(u_t^\epsilon)_j\tilde\gamma_{ik}(t,q_t^\epsilon))\frac{2}{\epsilon}\tilde K^\prime(\epsilon,t,q_t^\epsilon,\|u_t^\epsilon\|^2_A/\epsilon)A^{kl}(t,q_t^\epsilon)(u_t^\epsilon)_ldt\notag\\
&-(u_t^\epsilon)_i(\partial_{q^j}K)^\epsilon(t,x_t^\epsilon)dt-(u_t^\epsilon)_j(\partial_{q^i}K)^\epsilon(t,x_t^\epsilon)dt\notag\\
&+(u_t^\epsilon)_i(-\partial_t\psi_j(t,q_t^\epsilon)-\partial_{q^j}V(t,q_t^\epsilon)+F_j(t,x^\epsilon_t))dt\notag\\
&+(u_t^\epsilon)_j(-\partial_t\psi_i(t,q_t^\epsilon)-\partial_{q^i}V(t,q_t^\epsilon)+F_i(t,x^\epsilon_t))dt\notag\\
&+(u_t^\epsilon)_i\sigma_{j\rho}(t,x_t^\epsilon)dW^\rho_t+(u_t^\epsilon)_j\sigma_{i\rho}(t,x_t^\epsilon)dW^\rho_t+\Sigma_{ij}(t,x_t^\epsilon)dt,\notag
\end{align}
where we employed the equation for $u_t^\epsilon$, \req{z_eq}, and  \req{z_dH_eq}. We  isolate the $u$-dependent terms that appear in $(u_t^\epsilon)_j\partial_{p_l}H^\epsilon(t,x_t^\epsilon)dt$ to find
\begin{align}\label{Lyapunov_eq}
&\frac{2}{\epsilon}\tilde K^\prime(\epsilon,t,q_t^\epsilon,\|u_t^\epsilon\|^2_A/\epsilon)(\tilde\gamma_{jk}(t,q_t^\epsilon)A^{kl}(t,q_t^\epsilon)(u_t^\epsilon)_l(u_t^\epsilon)_i\\
&\hspace{35mm}+\tilde\gamma_{ik}(t,q_t^\epsilon)A^{kl}(t,q_t^\epsilon)(u_t^\epsilon)_l(u_t^\epsilon)_j)dt\notag\\
=&-d((u_t^\epsilon)_i(u_t^\epsilon)_j)-(u_t^\epsilon)_i(\partial_{q^j}K)^\epsilon(t,x_t^\epsilon)dt-(u_t^\epsilon)_j(\partial_{q^i}K)^\epsilon(t,x_t^\epsilon)dt\notag\\
&+(u_t^\epsilon)_i(-\partial_t\psi_j(t,q_t^\epsilon)-\partial_{q^j}V(t,q_t^\epsilon)+F_j(t,x^\epsilon_t))dt\notag\\
&+(u_t^\epsilon)_j(-\partial_t\psi_i(t,q_t^\epsilon)-\partial_{q^i}V(t,q_t^\epsilon)+F_i(t,x^\epsilon_t))dt\notag\\
&+(u_t^\epsilon)_i\sigma_{j\rho}(t,x_t^\epsilon)dW^\rho_t+(u_t^\epsilon)_j\sigma_{i\rho}(t,x_t^\epsilon)dW^\rho_t+\Sigma_{ij}(t,x_t^\epsilon)dt.\notag
\end{align}
We will solve this equation for  $\tilde K^\prime(\epsilon,t,q_t^\epsilon,\|u_t^\epsilon\|_A^2/\epsilon)(u_t^\epsilon)_j(u_t^\epsilon)_kdt$ using a Lyapunov equation technique,  as in  \cite{Hottovy2014,particle_manifold_paper}.

The formula for the left-hand side of \req{Lyapunov_eq} clearly represents a differential of a $C^1$-function. Therefore the integral from $0$ to $t$ of the right-hand side, which we denote by $(C_t)_{ij}$, is a $C^1$-function  $P$-a.s.  Differentiating both sides with respect to $t$, we obtain
\begin{align}\label{lyap_eq1}
&\frac{2}{\epsilon}\tilde K^\prime(\epsilon,t,q_t^\epsilon,\|u_t^\epsilon\|^2_A/\epsilon) (A\tilde\gamma )_j^l(t,q^\epsilon_t)(u_t^\epsilon)_l(u_t^\epsilon)_i\\
&+\frac{2}{\epsilon}\tilde K^\prime(\epsilon,t,q_t^\epsilon,\|u_t^\epsilon\|^2_A/\epsilon)(A\tilde\gamma )_i^l(t,q^\epsilon_t)(u_t^\epsilon)_l(u_t^\epsilon)_j =(\dot{C}_t)_{ij},\notag
\end{align}
where we define $(A\tilde\gamma)^i_j=\tilde\gamma_{jk}A^{ki}$.

Defining the matrix 
\begin{equation}
(V_t)_{ij}=\frac{2}{\epsilon}\tilde K^\prime(\epsilon,t,q_t^\epsilon,\|u_t^\epsilon\|^2_A/\epsilon) (u_t^\epsilon)_i(u_t^\epsilon)_j
\end{equation}
we rewrite \req{lyap_eq1} as
\begin{align}\label{lyap_eq_def}
(A\tilde\gamma)_i^l V_{lj}+V_{il}(A\tilde\gamma)^l_j=\dot{C}_{ij}.
\end{align}
This is a Lyapunov equation for $V$.

For every $T>0$, there exists $c>0$ and $\lambda>0$ such that $-A\tilde\gamma$ has  eigenvalues with real parts bounded above by $-c\lambda$, uniformly on $[0,T]\times\mathbb{R}^n$. See Lemma \ref{eig_bound_lemma2}.   Hence, we can solve uniquely for $V$,
\begin{align}\label{Lyap_sol}
V_{ij}=\int_0^\infty (e^{-yA\tilde\gamma})_i^k\dot{C}_{kl} (e^{-yA\tilde\gamma})^l_jdy.
\end{align}
 See, for example, Theorem 6.4.2 in \cite{ortega2013matrix}.

\begin{remark}
 Assumption \ref{assump5}, or something else that accomplishes a similar purpose, is necessary in the above computation.  If one tries to solve for $(\tilde V_t)^i_j\equiv (u_t^\epsilon)_j\partial_{p_i}H^\epsilon(t,x_t^\epsilon)$ directly, then one is led to the linear equation
\begin{align}\label{bad_V_eq}
\tilde{\gamma}_{jk}\tilde{V}^k_i+\tilde{\gamma}_{ik}\tilde{V}^k_j=\dot{C}_{ij}.
\end{align}
The left-hand side of this equation  has a non-trivial kernel, consisting of all $\tilde V$ for which $\tilde{\gamma}_{ik}\tilde{V}^k_j$ is antisymmetric.  Therefore, just knowing that $\tilde V$ satisfies \req{bad_V_eq} does not allow us to uniquely solve for $\tilde V$. Some additional constraint must be combined with \req{bad_V_eq} in order to solve for $\tilde V$.
\end{remark}

  Integrating \req{Lyap_sol} with respect to time, we obtain  
\begin{align}
&\frac{2}{\epsilon}\int_0^t\tilde K^\prime(\epsilon,s,q_s^\epsilon,\|u_s^\epsilon\|^2_A/\epsilon) (u_s^\epsilon)_i(u_s^\epsilon)_jds\\
=&\int_0^t \int_0^\infty \left(e^{-y(A\tilde\gamma)(s,q^\epsilon_s)}\right)^{k}_i  \left(e^{-y(A\tilde\gamma)(s,q^\epsilon_s)}\right)^{l}_j dy (\dot{C}_s)_{kl}ds.\notag
\end{align}
The functions
\begin{equation}\label{G_def}
G_{ij}^{kl}(t,q)=\int_0^\infty (e^{-y (A\tilde\gamma)(t,q)})_i^k (e^{-y(A\tilde\gamma)(t,q)})_j^l dy
\end{equation}
are $C^1$, hence $G_{ij}^{kl}(t,q^\epsilon_t)$ are  semimartingales and
\begin{align}\label{p2_eps_eq}
&\frac{2}{\epsilon}\tilde K^\prime(\epsilon,t,q_t^\epsilon,\|u_t^\epsilon\|^2_A/\epsilon) (u_t^\epsilon)_i (u_t^\epsilon)_jdt=G_{ij}^{ab}(t,q_t^\epsilon) d(C_t)_{ab}\\
=&G_{ij}^{ab}(t,q_t^\epsilon)\Sigma_{ab}(t,x_t^\epsilon)dt-G_{ij}^{ab}(t,q_t^\epsilon)d((u_t^\epsilon)_a(u_t^\epsilon)_b)\notag\\
&-G_{ij}^{ab}(t,q_t^\epsilon)(u_t^\epsilon)_a(\partial_{q^b}K)^\epsilon(t,x_t^\epsilon)dt-G_{ij}^{ab}(t,q_t^\epsilon)(u_t^\epsilon)_b(\partial_{q^a}K)^\epsilon(t,x_t^\epsilon)dt\notag\\
&+G_{ij}^{ab}(t,q_t^\epsilon)(u_t^\epsilon)_a(-\partial_t\psi_b(t,q_t^\epsilon)-\partial_{q^b}V(t,q_t^\epsilon)+F_b(t,x^\epsilon_t))dt\notag\\
&+G_{ij}^{ab}(t,q_t^\epsilon)(u_t^\epsilon)_b(-\partial_t\psi_a(t,q_t^\epsilon)-\partial_{q^a}V(t,q_t^\epsilon)+F_a(t,x^\epsilon_t))dt\notag\\
&+G_{ij}^{ab}(t,q_t^\epsilon)(u_t^\epsilon)_a\sigma_{b\rho}(t,x_t^\epsilon)dW^\rho_t+G_{ij}^{ab}(t,q_t^\epsilon)(u_t^\epsilon)_b\sigma_{a\rho}(t,x_t^\epsilon)dW^\rho_t.\notag
\end{align}

Combining \req{q_eq1} with \req{dq_K_eq}, \req{z_dH_eq}, and \req{p2_eps_eq} we see that $q_t^\epsilon$ satisfies the equation

\begin{align}\label{q_eq2}
d(q_t^\epsilon)^i=&(\tilde\gamma^{-1})^{ij}(t,q_t^\epsilon)(-\partial_t\psi_j(t,q_t^\epsilon)-\partial_{q^j}V(t,q_t^\epsilon)+F_j(t,x^\epsilon_t))dt\\
&+(\tilde\gamma^{-1})^{ij}(t,q_t^\epsilon)\sigma_{j\rho}(t,x_t^\epsilon)dW^\rho_t-(\tilde\gamma^{-1})^{ij}(t,q_t^\epsilon)\partial_{q^j}\tilde K(\epsilon,t,q_t^\epsilon,\|u_t^\epsilon\|^2_A/\epsilon)dt\notag\\
&+Q^{ikl}(t,q_t^\epsilon)J_{kl}(t,x_t^\epsilon)dt+d(R^\epsilon_t)^i,\notag
\end{align}
where
\begin{align}\label{J_def}
J_{ij}(t,x)\equiv G_{ij}^{kl}(t,q)\Sigma_{kl}(t,x),
\end{align}
\begin{align}\label{K_def}
 Q^{ijl}(t,q)\equiv \partial_{q^k}(\tilde\gamma^{-1})^{ij}(t,q) A^{kl}(t,q)-\frac{1}{2}(\tilde\gamma^{-1})^{ik}(t,q)\partial_{q^k} A^{jl}(t,q),
\end{align}
and
\begin{align}\label{R_def}
d(R^\epsilon_t)^i\equiv&- d((\tilde\gamma^{-1})^{ij}(t,q_t^\epsilon)(u^\epsilon_t)_j)+(u_t^\epsilon)_j\partial_t(\tilde\gamma^{-1})^{ij}(t,q_t^\epsilon)dt\\
&-Q^{ikl}(t,q_t^\epsilon)G_{kl}^{ab}(t,q_t^\epsilon)d((u_t^\epsilon)_a(u_t^\epsilon)_b)\notag\\
&+Q^{ikl}(t,q_t^\epsilon)G_{kl}^{ab}(t,q_t^\epsilon)(u_t^\epsilon)_a(-\partial_t\psi_b(t,q_t^\epsilon)-\partial_{q^b}V(t,q_t^\epsilon\notag)\\
&\hspace{48mm}-(\partial_{q^b}K)^\epsilon(t,x_t^\epsilon)+F_b(t,x^\epsilon_t))dt\notag\\
&+Q^{ikl}(t,q_t^\epsilon)G_{kl}^{ab}(t,q_t^\epsilon)(u_t^\epsilon)_b(-\partial_t\psi_a(t,q_t^\epsilon)-\partial_{q^a}V(t,q_t^\epsilon)\notag\\
&\hspace{48mm}-(\partial_{q^a}K)^\epsilon(t,x_t^\epsilon)+F_a(t,x^\epsilon_t))dt\notag\\
&+Q^{ikl}(t,q_t^\epsilon)G_{kl}^{ab}(t,q_t^\epsilon)(u_t^\epsilon)_a\sigma_{b\rho}(t,x_t^\epsilon)dW^\rho_t\notag\\
&+Q^{ikl}(t,q_t^\epsilon)G_{kl}^{ab}(t,q_t^\epsilon)(u_t^\epsilon)_b\sigma_{a\rho}(t,x_t^\epsilon)dW^\rho_t.\notag
\end{align}

Based on our knowledge of the decay rate of $u_t^\epsilon$, we expect $R^\epsilon_t$ to go to zero in the limit $\epsilon\rightarrow 0^+$.  In general, one  would still need to extract the portion of $(\tilde\gamma^{-1})^{ij}(t,q^\epsilon_t)\partial_{q^j}  \tilde K(\epsilon,t,q^\epsilon_t,\|z\|_A^2/\epsilon )dt$ that survives in the limit.  We will address this question in a future work \cite{ChucksVolumePaper}, but in this paper we will assume:
\begin{assumption}\label{assump6}
$\tilde K=\tilde K(\epsilon,t,z)$ i.e.  $\tilde K$ is independent of $q$, and hence 
\begin{align}
(\tilde\gamma^{-1})^{ij}(t,q^\epsilon_t)\partial_{q^j}  \tilde K(\epsilon,t,q^\epsilon_t,\|z\|_A^2/\epsilon )dt=0.
\end{align}
\end{assumption}

Along with Lemma \ref{p_decay_lemma}, the above calculations motivate the proposed limiting equation
\begin{align}\label{limit_eq1}
dq_t^i=&(\tilde\gamma^{-1})^{ij}(t,q_t)(-\partial_t\psi_j(t,q_t)-\partial_{q^j}V(t,q_t)+F_j(t,q_t,\psi(t,q_t)))dt\\
&+Q^{ikl}(t,q_t)J_{kl}(t,q_t,\psi(t,q_t))dt+(\tilde\gamma^{-1})^{ij}(t,q_t)\sigma_{j\rho}(t,q_t,\psi(t,q_t))dW^\rho_t.\notag
\end{align}
Note that an additional {\em noise induced drift} term, 
\begin{equation}\label{noise_induced_drift}
S^i(t,q)\equiv Q^{ijl}(t,q)J_{jl}(t,q,\psi(t,q)),
\end{equation}
 arises in the limit when $\Sigma$ is nonzero and (generally) when $\tilde\gamma$ and/or $A$ have nontrivial $q$-dependence.  This is in addition to the  forcing term, $-\partial_t\psi-\nabla_q V+F$, and is another manifestation of the phenomenon derived in \cite{Hottovy2014,particle_manifold_paper}.
\begin{remark}
Assumption \ref{assump6} determines the splitting of the Hamiltonian into $K^\epsilon(t,x)$ and $V(t,q)$, up to a function of time i.e. if $H^\epsilon=K^\epsilon_1+V_1=K^\epsilon_2+V_2$ are two splittings then $V_1(t,q)=V_2(t,q)+c(t)$.  This ambiguity does not impact the limiting equation \req{limit_eq1}, and so the limiting equation is uniquely defined by the original SDE, \req{Hamiltonian_SDE_q}-\req{Hamiltonian_SDE_p}, as it has to be, of course.
\end{remark}

\section{Convergence Proof}\label{sec:conv_proof}
In this final section, we prove convergence of $q_t^\epsilon$ to the solution of the proposed limiting equation, \req{limit_eq1}.  This will be accomplished by using the following lemma.

\begin{lemma}\label{conv_lemma}
Let $T>0$ and suppose we have continuous functions $\tilde F(t,x):[0,\infty)\times\mathbb{R}^{n}\times\mathbb{R}^n\rightarrow\mathbb{R}^n$, $\tilde \sigma(t,x):[0,\infty)\times\mathbb{R}^{n}\times\mathbb{R}^n\rightarrow\mathbb{R}^{n\times k}$, and  $\psi:[0,\infty)\times\mathbb{R}^n\rightarrow\mathbb{R}^n$ that are Lipschitz in $x$, uniformly in $t\in[0,T]$.

Let $W_t$ be a $k$-dimensional Wiener process, $p\geq 2$ and $\delta>0$ and suppose that we have continuous semimartingales $q_t$ and, for each $0<\epsilon\leq\epsilon_0$, $\tilde R_t^\epsilon$, $x_t^\epsilon=(q_t^\epsilon,p_t^\epsilon)$ that satisfy the following properties:
\begin{enumerate}
\item $q_t^\epsilon=q_0^\epsilon+\int_0^t\tilde F(s,x_s^\epsilon)ds+\int_0^t\tilde\sigma(s,x_s^\epsilon)dW_s+\tilde R^\epsilon_t$\label{q_eps_eq_assump}
\item $q_t=q_0+\int_0^t\tilde F(s,q_s,\psi(s,q_s))ds+\int_0^t\tilde\sigma(s,q_s^\epsilon,\psi(s,q_s))dW_s$\label{q_eq_assump}
\item $E[\|q_0^\epsilon-q_0\|^p]=O(\epsilon^\delta)\text{ as }\epsilon\rightarrow 0^+$. \label{ics_assump}
\item $E\left[\sup_{t\in[0,T]}\|\tilde R_t^\epsilon\|^p\right]=O(\epsilon^\delta)\text{ as }\epsilon\rightarrow 0^+$.\label{R_decay_assump}
\item $\sup_{t\in[0,T]}E[\|p_t^\epsilon-\psi(t,q_t^\epsilon)\|^p]=O(\epsilon^\delta)\text{ as }\epsilon\rightarrow 0^+$.\label{p_decay_assump}
\item $E\left[\sup_{t\in[0,T]}\|q_t^\epsilon\|^p\right]<\infty$ for all $\epsilon>0$ sufficiently small.\label{q_eps_integrable_assump}
\item $E\left[\sup_{t\in[0,T]}\|q_t\|^p\right]<\infty$\label{q_integrable_assump}

\end{enumerate}

Then 
\begin{align}
E\left[\sup_{t\in[0,T]}\|q_t^\epsilon-q_t\|^p\right]=O(\epsilon^\delta)\text{ as }\epsilon\rightarrow 0^+.
\end{align}

If we  replace properties $4$, $6$, $7$ with
\begin{enumerate}
\item[4*] $\sup_{t\in[0,T]}E\left[\|\tilde R_t^\epsilon\|^p\right]=O(\epsilon^\delta)\text{ as }\epsilon\rightarrow 0^+$,\label{R_decay_assump2}
\item[6*]  $\sup_{t\in[0,T]}E\left[\|q_t^\epsilon\|^p\right]<\infty$ for all $\epsilon>0$ sufficiently small,\label{q_eps_integrable_assump2}
\item[7*] $\sup_{t\in[0,T]}E\left[\|q_t\|^p\right]<\infty$,\label{q_integrable_assump2}
\end{enumerate}
then we instead arrive at
\begin{align}\label{sup_E_q}
\sup_{t\in[0,T]}E\left[\|q_t^\epsilon-q_t\|^p\right]=O(\epsilon^\delta)\text{ as }\epsilon\rightarrow 0^+.
\end{align}

\end{lemma}

\begin{proof}
The equations for $q_t^\epsilon$ and $q_t$ imply
\begin{align}
&(q_t^\epsilon)^i-(q_t)^i=(q_0^\epsilon)^i-(q_0)^i+\int_0^t\tilde F^i(s,x^\epsilon_s)-\tilde F^i(s,q_s,\psi(s,q_s))ds\\
&+\int_0^t\left[\tilde\sigma^i_{\rho}(s,x^\epsilon_s)-\tilde\sigma^i_\rho(s,q_s,\psi(s,q_s))\right] dW^\rho_s+(\tilde R^\epsilon_t)^i.\notag
\end{align}

 Using properties \ref{ics_assump} and \ref{R_decay_assump}, along with the H\"older and Burkholder-Davis-Gundy inequalities, for $0\leq t\leq T$ we have
\begin{align}
&E\left[\sup_{s\in[0,t]}\|q_s^\epsilon-q_s\|^p\right]\\
\leq& 4^{p-1}\bigg( E[\|q_0^\epsilon -q_0\|^p]+E\left[\left(\int_0^t\|\tilde F(s,x^\epsilon_s)-\tilde F (s,q_s,\psi(s,q_s))\|ds\right)^p\right]\notag\\
&+E\left[\sup_{s\in[0,t]}\left\|\int_0^s\tilde\sigma^i_\rho(r,x^\epsilon_r)-\tilde\sigma^i_{\rho}(r,q_r,\psi(r,q_r)) dW^\rho_r\right\|^p\right]+E\left[\sup_{s\in[0,t]}\|\tilde R^\epsilon_s\|^p\right]\bigg)\notag\\
\leq& 4^{p-1}\bigg(T^{p-1}E\left[\int_0^t\|\tilde F(s,x^\epsilon_s)-\tilde F (s,q_s,\psi(s,q_s))\|^pds\right]\notag\\
&+\tilde CE\left[\left(\int_0^t\|\tilde\sigma(r,x^\epsilon_r)-\tilde\sigma(r,q_r,\psi(r,q_r))\|_F^2dr\right)^{p/2}\right]\bigg)+O(\epsilon^\delta)\notag\\
\leq & 4^{p-1} \bigg( T^{p-1}\int_0^tE[\|\tilde F(s,x^\epsilon_s)-\tilde F (s,q_s,\psi(s,q_s))\|^p]ds\notag\\
&+\tilde CT^{p/2-1}\int_0^tE[\|\tilde\sigma(r,x^\epsilon_r)-\tilde\sigma(r,q_r,\psi(s,q_s))\|_F^p]dr\bigg)+O(\epsilon^\delta).\notag
\end{align}
By assumption, $\tilde\sigma$, $\tilde F$, and $\psi$ are Lipschitz in $x$, uniformly on $[0,T]$.  Hence, using property \ref{p_decay_assump},
\begin{align}\label{delta_q_bound}
&E\left[\sup_{s\in[0,t]}\|q_s^\epsilon-q_s\|^p\right]\leq \tilde C \int_0^tE[\|q^\epsilon_s-q_s\|^p+\|p_s^\epsilon-\psi(s,q_s)\|^p]ds+O(\epsilon^\delta)\\
\leq& \tilde C \int_0^tE[\|q^\epsilon_s-q_s\|^p+\|p_s^\epsilon-\psi(s,q^\epsilon_s)\|^p+\|\psi(s,q^\epsilon_s)-\psi(s,q_s)\|^p]ds+O(\epsilon^\delta)\notag\\
\leq& \tilde C \int_0^tE[\|q^\epsilon_s-q_s\|^p]ds+\tilde C\sup_{t\in[0,T]}E[\|p_t^\epsilon-\psi(t,q^\epsilon_t)\|^p]+O(\epsilon^\delta)\notag\\
=& \tilde C \int_0^tE\left[\sup_{r\in[0,s]}\|q^\epsilon_r-q_r\|^p\right]ds+O(\epsilon^\delta)\notag
\end{align}
for all $0\leq t\leq T$, where the constants change from line to line and are all independent of $t$.

Properties \ref{q_eps_integrable_assump} and \ref{q_integrable_assump} imply that $E\left[\sup_{s\in[0,t]}\|q_s^\epsilon-q_s\|^p\right]\in L^1([0,T])$ for $\epsilon$ sufficiently small, and hence Gronwall's inequality applied to \req{delta_q_bound} gives
\begin{align}
E\left[\sup_{s\in[0,t]}\|q_s^\epsilon-q_s\|^p\right]\leq O(\epsilon^\delta) e^{\tilde C t}
\end{align}
for $0\leq t\leq T$.

The proof of \req{sup_E_q} under assumptions 1-3, 4*, 5, 6*, and 7* is almost identical.
\end{proof}

To prove that the hypotheses of  Lemma \ref{conv_lemma} hold, we will need the following assumption:
\begin{assumption}\label{assump7}
We assume that, for every $T>0$, $\nabla_q V$, $F$, and $\sigma$ are Lipschitz in $x$ uniformly in $t\in[0,T]$.  We also assume that $A$ and $\gamma$ are $C^2$, $\psi$ is $C^3$, and $\partial_t\psi$, $\partial_{q^i}\psi$, $\partial_{q^i}\partial_{q^j}\psi$, $\partial_t\partial_{q^i}\psi$, $\partial_t\partial_{q^j}\partial_{q^i}\psi$, $\partial_{q^l}\partial_{q^j}\partial_{q^i}\psi$, $\partial_t\gamma$, $\partial_{q^i} \gamma$, $\partial_t\partial_{q^j}\gamma$,  $\partial_{q^i}\partial_{q^j}\gamma$, $\partial_t A$, $\partial_{q^i} A$, $\partial_t \partial_{q^i}A$,  and $\partial_{q^i}\partial_{q^j} A$ are bounded on $[0,T]\times\mathbb{R}^{n}$ for every $T>0$.
\end{assumption}
Note that, combined with our prior assumptions, this implies $\tilde\gamma$, $\tilde\gamma^{-1}$, $\partial_t\tilde\gamma^{-1}$, $\partial_{q^i}\tilde\gamma^{-1}$, $\partial_t\partial_{q^j}\tilde\gamma^{-1}$, and  $\partial_{q^i}\partial_{q^j}\tilde\gamma^{-1}$ are bounded on compact $t$ intervals. Additionally, using the formula for the derivative of the matrix exponential found in \cite{exp_deriv}, one can prove that our assumptions also  imply that the $G^{ij}_{kl}$'s are bounded and Lipschitz in $q$, uniformly on compact $t$ intervals.

As a step towards using Lemma \ref{conv_lemma} to prove our convergence result, we now show that $R^\epsilon_t$ from \req{R_def} converges to zero in the appropriate sense.

\begin{lemma}\label{R_decay_lemma}
Under Assumptions \ref{assump1}-\ref{assump5}, \ref{assump7}, for any $p>0$, $T>0$, $0<\beta<p/2$ we have
\begin{align}
E\left[\sup_{t\in[0,T]}\|R_t^\epsilon\|^p\right]=O(\epsilon^\beta) \text{ as } \epsilon \rightarrow 0^+
\end{align}
and 
\begin{align}\label{sup_E_R}
\sup_{t\in[0,T]}E\left[\|R_t^\epsilon\|^p\right]=O(\epsilon^{p/2}) \text{ as } \epsilon \rightarrow 0^+,
\end{align}
where $R_t^\epsilon$ was defined in \req{R_def}.
\end{lemma}
\begin{proof}
Let us first assume that $p>2$.  Let  $0<\beta<p/2$. Define
\begin{align}
Y(t,x)=-\partial_t\psi(t,q)-\nabla_qV(t,q)+F(t,x).
\end{align}
Our assumptions imply that $Y$ is bounded on $[0,T]\times\mathbb{R}^{2n}$. 

 From \req{R_def},
\begin{align}
&E\left[\sup_{t\in[0,T]}\|R^\epsilon_t\|^p\right]\leq 8^{p-1}\bigg(E\left[\sup_{t\in[0,T]}\|(\tilde\gamma^{-1})^{ij}(t,q^\epsilon_t)({u}^\epsilon_t)_j\|^p\right]\\
&+E[\|(\tilde\gamma^{-1})^{ij}(0,q^\epsilon_0)({u}^\epsilon_0)_j\|^p]+E\left[\left(\int_0^T\|(u^\epsilon_s)_j \partial_s(\tilde\gamma^{-1})^{ij}(s,q^\epsilon_s)\|ds\right)^p\right]\notag\\
&+E\left[\sup_{t\in[0,T]}\left\|\int_0^tQ^{ijl}(s,q_s^\epsilon)G_{jl}^{ab}(s,q_s^\epsilon)d((u_s^\epsilon)_a(u_s^\epsilon)_b)\right\|^p\right]\notag\\
&+E\left[\left(\int_0^T\|Q^{ijl}(s,q_s^\epsilon)G_{jl}^{ab}(s,q_s^\epsilon)(u_s^\epsilon)_a(-(\partial_{q^b}K)^\epsilon(s,x_s^\epsilon)+Y_b(s,x^\epsilon_s))\|ds\right)^p\right]\notag\\
&+E\left[\left(\int_0^T\|Q^{ijl}(s,q_s^\epsilon)G_{jl}^{ab}(s,q_s^\epsilon)(u_s^\epsilon)_b(-(\partial_{q^a}K)^\epsilon(s,x_s^\epsilon)+Y_a(s,x^\epsilon_s))\|ds\right)^p\right]\notag\\
&+E\left[\sup_{t\in[0,T]}\left\|\int_0^tQ^{ijl}(s,q_s^\epsilon)G_{jl}^{ab}(s,q_s^\epsilon)(u_s^\epsilon)_a\sigma_{b\rho}(s,x_s^\epsilon) dW^\rho_s\right\|^p\right]\notag\\
&+E\left[\sup_{t\in[0,T]}\left\|\int_0^tQ^{ijl}(s,q_s^\epsilon)G_{jl}^{ab}(s,q_s^\epsilon)(u_s^\epsilon)_b\sigma_{a\rho}(s,x_s^\epsilon) dW^\rho_s\right\|^p\right]\bigg),\notag
\end{align}
where the norm is the 2-norm of vectors, with components indexed by $i$, resulting from summation over other, repeated indices.  We now show that all of these terms are $O(\epsilon^\beta)$.

Boundedness of $\tilde\gamma^{-1}$ together with Lemma \ref{p_decay_lemma}  implies that the first two terms satisfy
\begin{align}
&E\left[\sup_{t\in[0,T]}\|(\tilde\gamma^{-1})^{ij}(t,q^\epsilon_t)({u}^\epsilon_t)_j\|^p\right]+E[\|(\tilde\gamma^{-1})^{ij}(0,q^\epsilon_0)({u}^\epsilon_0)_j\|^p]\\
\leq &2\|\tilde\gamma^{-1}\|_\infty^pE\left[\sup_{t\in[0,T]}\|{u}^\epsilon_t\|^p\right]=O(\epsilon^\beta).\notag
\end{align}

By H\"older's inequality, boundedness of $\partial_t\tilde\gamma^{-1}$, and Lemma  \ref{p_decay_lemma} , the third term satisfies
\begin{align}
&E\left[\left(\int_0^T\|(u^\epsilon_s)_j \partial_t(\tilde\gamma^{-1})^{ij}(s,q^\epsilon_s)\|ds\right)^p\right]\\
\leq& T^{p-1}\|\partial_t\tilde\gamma^{-1}\|_\infty^p \int_0^TE[\|u^\epsilon_s\|^p]ds\notag\\
\leq &T^{p}\|\partial_t\tilde\gamma^{-1}\|_\infty^p \sup_{s\in[0,T]}E[\|u^\epsilon_s\|^p]=O(\epsilon^{p/2}).\notag
\end{align}

The functions $Q^{ijl}(t,q)G^{ab}_{jl}(t,q)$ are $C^1$, bounded, with bounded first derivatives on $[0,T]\times\mathbb{R}^n$.
  Therefore, by Proposition \ref{p_int_decay} we have
\begin{align}
E\left[\sup_{t\in[0,T]}\left\|\int_0^tQ^{ijl}(s,q_s^\epsilon)G_{jl}^{ab}(s,q_s^\epsilon)d((u_s^\epsilon)_a(u_s^\epsilon)_b)\right\|^p\right]=O(\epsilon^{p/2}).
\end{align}

Using H\"older's inequality, \req{K_assump1}, and our various boundedness assumptions, the fifth term can be bounded as follows:
\begin{align}
&E\left[\left(\int_0^T\|Q^{ijl}(s,q_s^\epsilon)G_{jl}^{ab}(s,q_s^\epsilon)(u_s^\epsilon)_a(-(\partial_{q^b}K)^\epsilon(s,x_s^\epsilon)+Y_b(s,x^\epsilon_s))\|ds\right)^p\right]\notag\\
\leq &\tilde C T^{p-1}E\left[\int_0^T\|u_s^\epsilon\|^p(\|(\partial_{q^b}K)^\epsilon(s,x_s^\epsilon)\|+\|Y\|_\infty)^pds\right]\notag\\
\leq &\tilde C T^{p}\sup_{s\in[0,T]}E\left[\|u_s^\epsilon\|^p(M+CK^\epsilon(s,x_s^\epsilon)+\|Y\|_\infty)^p\right].
\end{align}
Again, here and in the following, we let $\tilde C$ denote a constant that may vary from line to line. We  now use the Cauchy-Schwarz  inequality,   Lemma \ref{p_decay_lemma}, and Proposition \ref{Sup_E_prop} to obtain
\begin{align}
&E\left[\left(\int_0^T\|Q^{ijl}(s,q_s^\epsilon)G_{jl}^{ab}(s,q_s^\epsilon)(u_s^\epsilon)_a(-(\partial_{q^b}K)^\epsilon(s,x_s^\epsilon)+Y_b(s,x^\epsilon_s))\|ds\right)^p\right]\notag\\
\leq &\tilde C T^{p}\sup_{s\in[0,T]}E[\|u_s^\epsilon\|^{2p}]^{1/2}\sup_{s\in[0,T]}E[(M+CK^\epsilon(s,x_s^\epsilon)+\|Y\|_\infty)^{2p}]^{1/2}\\
=&O(\epsilon^p)^{1/2}O(1)=O(\epsilon^{p/2}).\notag
\end{align}
A similar argument shows that the sixth term is also $O(\epsilon^{p/2})$.

Using the  Burkholder-Davis-Gundy and H\"older inequalities together with the boundedness assumptions and Lemma \ref{p_decay_lemma}, the seventh term satisfies
\begin{align}
&E\left[\sup_{t\in[0,T]}\left\|\int_0^tQ^{ijl}(s,q_s^\epsilon)G_{jl}^{ab}(s,q_s^\epsilon)(u_s^\epsilon)_a\sigma_{b\rho}(s,x_s^\epsilon) dW^\rho_s\right\|^p\right]\\
\leq &\tilde CE\left[\left(\int_0^T\sum_{i,\rho}(Q^{ijl}(s,q_s^\epsilon)G_{jl}^{ab}(s,q_s^\epsilon)(u_s^\epsilon)_a\sigma_{b\rho}(s,x_s^\epsilon))^2ds\right)^{p/2}\right]\notag\\
\leq &\tilde C E\left[\left(\int_0^T\|u_s^\epsilon\|^2 ds\right)^{p/2}\right]\leq \tilde CT^{p/2-1} E\left[\int_0^T\|u_s^\epsilon\|^p ds \right]\notag\\
\leq &\tilde C\sup_{s\in[0,T]}E[\|u_s^\epsilon\|^p]=O(\epsilon^{p/2}),\notag
\end{align}
where the power of $T$ can be absorbed into the constant, since we are working on a fixed time interval.
A similar estimate applies to the final, eighth term. Therefore, we have proven the claim for $p>2$.  An application of H\"older's inequality proves it for all $p>0$.  The proof of \req{sup_E_R} is nearly identical.
\end{proof}

We now have all the ingredients to prove convergence of $q_t^\epsilon$ to $q_t$.
\begin{theorem}\label{conv_thm}
Let $x_t^\epsilon$ be a family of solutions to the SDE \ref{Hamiltonian_SDE_q}-\ref{Hamiltonian_SDE_p} with  initial condition $(q_0^\epsilon,p_0^\epsilon)$ and $q_t$ be a solution to the proposed limiting SDE, \req{limit_eq1}, with initial condition $q_0$.  Suppose that for all $\epsilon>0$ and  all $p>0$ we have $E[\|q^\epsilon_0\|^p]<\infty$, $E[\|q_0\|^p]<\infty$, and $E[\|q_0^\epsilon-q_0\|^p]=O(\epsilon^{p/2})$.  Also suppose that Assumptions \ref{assump1}-\ref{assump7} hold.  Then for any $T>0$, $p>0$, $0<\beta<p/2$ we have
\begin{align}\label{q_conv_rate}
E\left[\sup_{t\in[0,T]}\|q_t^\epsilon-q_t\|^p\right]=O(\epsilon^{\beta})\text{ as } \epsilon\rightarrow 0^+
\end{align}
and
\begin{align}\label{q_conv_rate2}
\sup_{t\in[0,T]}E\left[\|q_t^\epsilon-q_t\|^p\right]=O(\epsilon^{p/2})\text{ as } \epsilon\rightarrow 0^+.
\end{align}
\end{theorem}
\begin{proof}
First let $p>2$. Define the generalized force vector field 
\begin{equation}
\tilde F(t,x)^i=(\tilde\gamma^{-1})^{ij}(t,q)(-\partial_t\psi_i(t,q)-\partial_{q^j} V(t,q)+F_j(t,x))+Q^{ijl}(t,q)J_{jl}(t,x)
\end{equation}
and noise coefficients
\begin{align}
\tilde\sigma^i_\eta(t,x)=(\tilde\gamma^{-1})^{ij}(t,q)\sigma_{j\eta}(t,x).
\end{align}
Our assumptions imply that these are bounded on $[0,T]\times\mathbb{R}^{2n}$ and, along with $\psi$, they are Lipschitz in $x$, uniformly in $t\in [0,T]$.

We will now check all of the properties that are required to use Lemma \ref{conv_lemma}. By \req{q_eq2} and \req{limit_eq1}, $q_t^\epsilon$ and $q_t$ satisfy the equations
\begin{align}
&(q_t^\epsilon)^i=(q_0^\epsilon)^i+\int_0^t\tilde F^i(s,x^\epsilon_s)ds+\int_0^t\tilde \sigma^i_\eta(s,x^\epsilon_s) dW^\eta_s+(R^\epsilon_t)^i
\end{align}
and
\begin{align}\label{q_SDE}
&(q_t)^i=(q_0)^i+\int_0^t\tilde F^i(s,q_s,\psi(s,q_s))ds+\int_0^t\tilde\sigma^i_\eta(s,q_s,\psi(s,q_s)) dW^\eta_s.
\end{align}
Note, that the term involving $\partial_{q^j}\tilde{K}$, present in \req{q_eq2}, vanishes under Assumption \ref{assump6}. In addition, for each $T>0$, $\tilde F(t,q,\psi(t,q))$ and $\tilde\sigma(t,q,\psi(t,q))$ are bounded and Lipschitz in $q$, uniformly in $t$ on $[0,T]\times\mathbb{R}^n$, so a unique solution to \req{q_SDE} exists and is defined for all $t\geq 0$ \cite{khasminskii2011stochastic}.

By assumption, the initial conditions satisfy $E[\|q_0^\epsilon-q_0\|^p]=O(\epsilon^{p/2})=O(\epsilon^\beta)$.  Lemma \ref{R_decay_lemma} implies that 
\begin{equation}
E\left[\sup_{t\in[0,T]}\|R_t^\epsilon\|^p\right]=O(\epsilon^\beta), \hspace{2mm} \sup_{t\in[0,T]}E\left[\|R_t^\epsilon\|^p\right]=O(\epsilon^{p/2})  
\end{equation}
as $\epsilon \rightarrow 0^+$ and,  by Lemma \ref{p_decay_lemma},   
\begin{equation}
\sup_{t\in[0,T]}E[\|p_t^\epsilon-\psi(t,q_t^\epsilon)\|^p]=O(\epsilon^{p/2})=O(\epsilon^\beta)\text{ as }\epsilon\rightarrow 0^+.
\end{equation}

For any $\epsilon>0$, $p>1$ we have
\begin{align}
&\sup_{t\in[0,T]}E\left[\|q_t^\epsilon\|^p\right]\leq E\left[\sup_{t\in[0,T]}\|q_t^\epsilon\|^p\right]\\
=&E\left[\sup_{t\in[0,T]}\|q_0^\epsilon+\int_0^t \nabla_p H^\epsilon(r,x_r^\epsilon) dr\|^p\right]\notag\\
\leq &2^{p-1}E[\|q_0^\epsilon\|^p]+2^{p-1}E\left[\left(\int_0^T \|\nabla_p K^\epsilon(r,x_r^\epsilon)\| dr\right)^p\right]\notag\\
= &2^{p-1}E[\|q_0^\epsilon\|^p]+2^{p-1}E\left[\left(\int_0^T \epsilon^{-1/2}\|(\nabla_z K)^\epsilon(r,x_r^\epsilon)\| dr\right)^p\right]\notag\\
\leq&2^{p-1}E[\|q_0^\epsilon\|^p]+2^{p-1} T^{p-1} \epsilon^{-p/2}E\left[\int_0^T \left(M+K^\epsilon(r,x_r^\epsilon)\right)^pdr\right]\notag\\
\leq&2^{p-1}E[\|q_0^\epsilon\|^p]+4^{p-1} T^{p} \epsilon^{-p/2}\left(M^p+\sup_{r\in[0,T]}E\left[ K^\epsilon(r,x_r^\epsilon)^p\right]\right)<\infty,\notag
\end{align}
where we used Assumption \ref{assump1} to bound $\nabla_z K$ and  Proposition \ref{Sup_E_prop} in the last line.

We also have
\begin{align}
&\sup_{t\in[0,T]}E\left[\|q_t\|^p\right]\leq E\left[\sup_{t\in[0,T]}\|q_t\|^p\right]\\
\leq& 3^{p-1}\bigg(E[\|q_0\|^p]+E\left[\left(\int_0^{T}\|\tilde F(s,q_s,\psi(s,q_s))\|ds\right)^p\right]\notag\\
&+E\left[\sup_{t\in[0,T]}\left\|\int_0^{t}\tilde \sigma(s,q_s,\psi(s,q_s))dW_s\right\|^p\right]\bigg).\notag
\end{align}
$\tilde F$ and $\tilde\sigma$ are bounded uniformly up to time $T$, so applying the Burkholder-Davis-Gundy and H\"older  inequalities to the last term, we get
\begin{align}
\sup_{t\in[0,T]}E\left[\|q_t\|^p\right]\leq E\left[\sup_{t\in[0,T]}\|q_t\|^p\right]\leq& 3^{p-1}(E[\|q_0\|^p]+T^p\|\tilde F\|_\infty^p +\tilde CT^{p/2}\|\tilde \sigma\|_{F,\infty}^p)<\infty.
\end{align}

Hence we have verified properties 1-7 from Lemma  \ref{conv_lemma} for every $T>0$, $p>2$, and $\delta=\beta\in (0,p/2)$, and properties 1-3, 4${}^*$, 5, 6${}^*$, 7${}^*$ for every $T>0$, $p>2$, with $\delta=p/2$. We can therefore  conclude the convergence results \req{q_conv_rate} and \req{q_conv_rate2}.  The results for any $p>0$ follows from an application of H\"older's inequality.

\end{proof}

\section{Extension to Unbounded Forces}\label{sec:unbounded} 
In this section, we focus on relaxing some of the boundedness assumptions in Theorem \ref{conv_thm} by adapting the method developed in \cite{herzog2015small}.  We make no claim that the assumptions here are as weak or general as possible and there are many variations on this idea that one could pursue, weakening the boundedness assumptions on the various objects appearing in the SDEs; here we focus on accommodating  unbounded forces, $F$ and $\nabla_qV$. Specifically, we will be able to prove convergence for potentials that are confining, or at least not too unstable, meaning that there exists $a\geq 0, b\geq 0$ such that $a+b\|q\|^2+V(t,q)$ is non-negative.
\begin{theorem}\label{thm:conv_in_prob}
Suppose that the following hold:
\begin{enumerate}
\item The family of Hamiltonians has the form
\begin{align}
H^\epsilon(t,x)=K(\epsilon,t,q,(p-\psi(t,q))/\sqrt{\epsilon})+V(t,q)
\end{align}
where $V$ is a $C^2$, $\mathbb{R}$-valued function,  $\psi$ is a $C^3$, $\mathbb{R}^n$-valued function, and
\begin{align}
K(\epsilon,t,q,z)=\tilde K(\epsilon,t,A^{ij}(t,q)z_iz_j)
\end{align}
for a non-negative function $\tilde K(\epsilon,t,\zeta)$  that is $C^2$ in $(t,\zeta)\in [0,\infty)\times[0,\infty)$ and a $C^2$, positive definite  $n \times n$ matrix-valued function, $A$.
\item  For every $T>0$, the following bounds hold on $(0,\epsilon_0]\times[0,\infty)\times\mathbb{R}^{2n}$:
\begin{enumerate}
\item There exist $C>0$ and $M>0$ such that
\begin{align}
\max\{|\partial_t K(\epsilon,t,q,z)|,\|\nabla_q K(\epsilon,t,q,z)\|\}\leq M+CK(\epsilon,t,q,z).
\end{align}
\item There exist  $c>0$ and $M\geq 0$ such that
\begin{align}
\|\nabla_z K(\epsilon,t,q,z)\|^2+M\geq c K(\epsilon,t,q,z).
\end{align}
\item
For every $\delta>0$ there exists an $M>0$ such that
\begin{align}
&\max\left\{\|\nabla_z K(\epsilon,t,q,z)\|,\left(\sum_{i,j}|\partial_{z_i}\partial_{z_j}K(\epsilon,t,q,z)|^2\right)^{1/2}\right\}\\
\leq& M+\delta K(\epsilon,t,q,z).\notag
\end{align}
\item  There exists $c>0$, $\eta>0$ such that
\begin{align}
K(\epsilon,t,q,z)\geq c\|z\|^{2\eta}.
\end{align}
\end{enumerate}
\item There exists $a\geq 0, b\geq 0$ such that $\tilde V(t,q)\equiv a+b\|q\|^2+V(t,q)$ is non-negative.
\item $\gamma$ is $C^2$, independent of $p$, and symmetric with eigenvalues bounded below by some $\lambda>0$.
\item The eigenvalues of $A$ are bounded below by some $c>0$.
\item  $\sigma(t,x)$  and $F(t,x)$ are continuous and   Lipschitz in $x$ with the Lipschitz constant uniform on compact time intervals.
\item $\sigma$, $\gamma$, $\partial_{q^i}\psi$,  $A$, $\partial_{q^i}A$, $\partial_t A$, $\partial_{q^i} A$, $\partial_t \partial_{q^i}A$,  and $\partial_{q^i}\partial_{q^j} A$ are bounded.
\item There exists $C>0$, $M>0$ such that
\begin{align}
|\partial_t V(t,q)|\leq M+C(\|q\|^2+\tilde V(t,q)),
\end{align}
\begin{align}
\|\partial_t\psi(t,q)\|^2\leq M+C\left(\|q\|^2+\tilde V(t,q)\right),
\end{align}
\begin{align}
\|F(t,x)\|^2\leq M+C\left(\|q\|^2+\tilde V(t,q)\right),
\end{align}
\begin{align}
\|\partial_{q^i}\tilde\gamma(t,q)\|\leq M+C\left(\|q\|^2+\tilde V(t,q)\right), \hspace{2mm} i=1,...,n,
\end{align}
and
\begin{align}
\left(\sum_{i,j}|\partial_{q^i}\partial_{q^j} V(t,q)|^2\right)^{1/2}\leq M+C\left(\|q\|^2+\tilde V(t,q)\right).
\end{align}
\item We have $\mathbb{R}^n$-valued initial conditions  $x^\epsilon_0=(q_0^\epsilon,p_0^\epsilon)$ and $q_0$ that satisfy the following:

For some $R>0$, $C>0$ we have $\|q_0\|\leq R$, $\|q_0^\epsilon\|\leq R$ and  $K^\epsilon(0,x^\epsilon_0)\leq C$ for all $\epsilon>0$ and all $\omega\in\Omega$.

For all $p>0$ we have   and $E[\|q_0^\epsilon-q_0\|^p]=O(\epsilon^{p/2})$.

\end{enumerate}

Let $x_t^\epsilon$ be the family of solutions to the SDE \ref{Hamiltonian_SDE_q}-\ref{Hamiltonian_SDE_p} with  initial condition $x_0^\epsilon$ and $q_t$ be a solution to the proposed limiting SDE, \req{limit_eq1}, with initial condition $q_0$.  For any $T>0$, $\delta>0$, we have
\begin{align}
\lim_{\epsilon\to 0^+}P\left(\sup_{t\in[0,T]}\|q_t^\epsilon-q_t\|>\delta\right)=0.
\end{align}
\end{theorem}
\begin{proof}
By Lemmas \ref{unique_lemma1} and \ref{unique_lemma2}, the maximal solutions, $x_t^\epsilon$, to the SDE \req{Hamiltonian_SDE_q}-\req{Hamiltonian_SDE_p} and  $q_t$ to the SDE \req{limit_eq_intro} are unique a.s. and a.s. exist for all $t\geq 0$.

Let $\chi:\mathbb{R}^n\to [0,1]$ be a $C^\infty$ bump function, equal to $1$ on $B_1(0)\equiv\{\|q\|\leq 1\}$ and zero outside $B_2(0)$. Given $r>0$ let $\chi_r(q)=\chi(q/r)$.  Define
\begin{align}
&V_r(t,q)=\chi_r(q)V(t,q), \hspace{2mm} F_r(t,x)=\chi_r(q) F(t,x), \hspace{2mm}  \psi_r(t,q)=\chi_r(q)\psi(t,q),\notag\\
&\gamma_r(t,q)=\chi_r(q)\gamma(t,q)+(1-\chi_r(q))\lambda I.
\end{align}
Replacing $V$ with $V_r$, $F$ with $F_r$ etc., we arrive at an SDE satisfying the hypotheses of Theorem \ref{conv_thm}.  Let $x_t^{r,\epsilon}$ be the corresponding solution to \req{Hamiltonian_SDE_q}-\req{Hamiltonian_SDE_p}  and $q_t^r$ the limit of $q_r^{r,\epsilon}$, both using the same initial conditions as the original systems.   Theorem \ref{conv_thm} then implies that for all  $T>0$, $p>0$, $0<\beta<p/2$ and all $r>R$ we have
\begin{align}\label{qr_conv}
E\left[\sup_{t\in[0,T]}\|q_t^{r,\epsilon}-q^r_t\|^p\right]=O(\epsilon^{\beta})\text{ as } \epsilon\rightarrow 0^+.
\end{align}
We will now use this result to prove that $q_t^\epsilon$ converges to $q_t$ in probability.

 For each $r>R$ define the stopping times $\tau^\epsilon_r=\inf\{t:\|q_t^\epsilon\|\geq r\}$ and $\tau_r=\inf\{t:\|q_t\|\geq r\}$.   The drifts and diffusions of the modified and unmodified SDEs agree on the ball $\{\|q\|\leq r\}$, so uniqueness of solutions implies 
\begin{align}\label{q_eps_unique}
q^\epsilon_{\tau^\epsilon_r\wedge t}=q^{r,\epsilon}_{\tau^\epsilon_r\wedge t} \text{ for all $t\geq 0$ a.s. }
\end{align}
 and  
\begin{align}
q_{\tau_r\wedge t}= q^r_{\tau_{ r}\wedge t}  \text{ for all $t\geq 0$ a.s.}
\end{align}

Using this, given $T>0$, $\delta>0$ we can calculate
\begin{align}
&P\left(\sup_{t\in[0,T]}\|q_t^\epsilon-q_t\|>\delta\right)\\
=&P\left(\tau_r\wedge\tau^\epsilon_r> T,\sup_{t\in[0,T]}\|q_{\tau^\epsilon_r\wedge t}^{\epsilon}-q_{\tau_r\wedge t}\|>\delta\right)+P\left(\tau_r\wedge\tau^\epsilon_r\leq T,\sup_{t\in[0,T]}\|q_t^\epsilon-q_t\|>\delta\right)\notag\\
=&P\left(\tau_r\wedge\tau^\epsilon_r> T,\sup_{t\in[0,T]}\|q_t^{r,\epsilon}-q^r_t\|>\delta\right)+P\left(\tau_r\wedge\tau^\epsilon_r\leq T,\sup_{t\in[0,T]}\|q_t^\epsilon-q_t\|>\delta\right)\notag\\
\leq&P\left(\sup_{t\in[0,T]}\|q_t^{r,\epsilon}-q^r_t\|>\delta\right)+P\left(\tau_r\wedge\tau^\epsilon_r\leq T\right)\notag.
\end{align}
The first term converges to zero as $\epsilon\to 0^+$ by \req{qr_conv}, so we focus on the second.
\begin{align}
P\left(\tau_r\wedge\tau^\epsilon_r\leq T\right)\leq &P\left(\sup_{t\in [0,T]}\|q_t^{r,\epsilon}-q_t^{r}\|> 1\right)+P\left(\tau_r\leq T\right)\\
&+P\left(\tau_r>T,\tau^\epsilon_r\leq T, \sup_{t\in [0,T]}\|q_t^{r,\epsilon}-q_t^{r}\|\leq 1\right)\notag\\
\leq &P\left(\sup_{t\in [0,T]}\|q_t^{r,\epsilon}-q_t^{r}\|> 1\right)+P\left(\sup_{t\in[0,T]}\|q_t\|\geq r\right)\notag\\
&+P\left(\tau^\epsilon_r\leq T, \|q_{\tau^\epsilon_r\wedge T}^{r,\epsilon}-q_{\tau^\epsilon_r\wedge T}\|\leq 1\right)\notag\\
=& P\left(\sup_{t\in [0,T]}\|q_t^{r,\epsilon}-q_t^{r}\|> 1\right)+P\left(\sup_{t\in[0,T]}\|q_t\|\geq r\right)\notag\\
&+P\left(\tau^\epsilon_r\leq T, \|q_{\tau^\epsilon_r\wedge T}^{\epsilon}-q_{\tau^\epsilon_r\wedge T}\|\leq 1\right),\notag
\end{align}
where we again used the uniqueness result, \req{q_eps_unique}.

On the event where $\tau^\epsilon_r\leq T$ and $\|q_{\tau^\epsilon_r\wedge T}^{\epsilon}-q_{\tau^\epsilon_r\wedge T}\|\leq 1$  we have $\|q^\epsilon_{\tau^\epsilon_r\wedge T}\|\geq r$ and hence 
\begin{align}
\|q_{ \tau^\epsilon_r\wedge T}\|\geq \|q^\epsilon_{\tau^\epsilon_r\wedge T}\|-\|q^\epsilon_{\tau^\epsilon_r\wedge T}-q_{\tau^\epsilon_r\wedge T}\|\geq r-1.
\end{align}
Therefore $\sup_{t\in[0,T]}\|q_t\|\geq r-1$ on this event.  

Combining the above calculation with \req{qr_conv}, for any $r>0$ we obtain 
\begin{align}
&\limsup_{\epsilon\to 0^+}P\left(\sup_{t\in[0,T]}\|q_t^\epsilon-q_t\|>\delta\right)\\
\leq&\limsup_{\epsilon\to 0^+}P\left(\sup_{t\in[0,T]}\|q_t^{r,\epsilon}-q^r_t\|>\delta\right)+ \limsup_{\epsilon\to 0^+}P\left(\sup_{t\in [0,T]}\|q_t^{r,\epsilon}-q_t^{r}\|> 1\right)\notag\\
&+P\left(\sup_{t\in[0,T]}\|q_t\|\geq r\right)+P\left(\sup_{t\in[0,T]}\|q_t\|\geq r-1\right)\notag\\
\leq& 2P\left(\sup_{t\in[0,T]}\|q_t\|\geq r-1\right).\notag
\end{align}

Non-explosion of $q_t$ implies that $P\left(\sup_{t\in[0,T]}\|q_t\|\geq r-1\right)\to 0$ as $r\to\infty$ and so we have proven the claimed result.

\end{proof}

\appendix

\section{Assumptions} \label{app:assump}
In this appendix we collect all the assumptions that are needed for one of the main results, Theorem \ref{conv_thm}.  In the body of the paper, we will restate each assumption when it is first used.
\setcounter{assumption}{0}
\begin{assumption}
We assume that $\sigma$, $F$, and $\gamma$ are continuous, the Hamiltonian has the form given in \req{H_family} where $K(\epsilon,t,q,z)$ is non-negative and $C^2$ in $(t,q,z)$ for each $\epsilon$, $\psi$ is  $C^2$, and the solutions, $x^\epsilon_t$, to the SDE \ref{Hamiltonian_SDE_q}-\ref{Hamiltonian_SDE_p} exist for all $t\geq 0$.

  For every $T>0$, we assume the following bounds hold on $(0,\epsilon_0]\times [0,T]\times\mathbb{R}^{2n}$:
\begin{enumerate}
\item There exist $C>0$ and $M>0$ such that
\begin{align}
\max\{|\partial_t K(\epsilon,t,q,z)|,\|\nabla_q K(\epsilon,t,q,z)\|\}\leq M+CK(\epsilon,t,q,z).
\end{align}
\item There exist  $c>0$ and $M\geq 0$ such that
\begin{align}
\|\nabla_z K(\epsilon,t,q,z)\|^2+M\geq c K(\epsilon,t,q,z).
\end{align}
\item
For every $\delta>0$ there exists an $M>0$ such that
\begin{align}
\max\left\{\|\nabla_z K(\epsilon,t,q,z)\|,\left(\sum_{i,j}|\partial_{z_i}\partial_{z_j}K(\epsilon,t,q,z)|^2\right)^{1/2}\right\}\leq M+\delta K(\epsilon,t,q,z).
\end{align}
\end{enumerate}

\end{assumption}

\begin{assumption}
For every $T>0$, the following hold uniformly on $[0,T]\times\mathbb{R}^n$:
\begin{enumerate}
\item $V$ is $C^2$ and $\nabla_q V$ is bounded.
\item  $\gamma$ is symmetric with eigenvalues bounded below by some $\lambda>0$.
\item  $\gamma$, $F$, $\partial_t\psi$, and $\sigma$ are bounded.
\item There exists $C>0$ such that the (random) initial conditions satisfy $K^\epsilon(0,x^\epsilon_0)\leq C$ for all $\epsilon>0$ and all $\omega\in\Omega$.
\end{enumerate}
\end{assumption}

\begin{assumption}
For every $T>0$ there exists $c>0$, $\eta>0$ such that
\begin{align}
K(\epsilon,t,q,z)\geq c\|z\|^{2\eta}
\end{align}
on $(0,\epsilon_0]\times[0,T]\times\mathbb{R}^{2n}$.
\end{assumption}

\begin{assumption}
 $\gamma$ is $C^1$ and is independent of $p$.
\end{assumption}

\begin{assumption}\label{app:assump5}
 $K$ has the form
\begin{align}
K(\epsilon,t,q,z)=\tilde K(\epsilon,t,q,A^{ij}(t,q)z_iz_j)
\end{align}
where $\tilde K(\epsilon,t,q,\zeta)$ is $C^2$ in $(t,q,\zeta)$ for every $\epsilon$, non-negative on $(0,\epsilon_0]\times[0,\infty)\times\mathbb{R}^n\times[0,\infty)$ and $A(t,q)$ is a $C^2$ function whose values are symmetric $n \times n$-matrices.   We also assume that for every $T>0$, the eigenvalues of $A$ are bounded above and below by some constants $C>0$ and $c>0$ respectively, uniformly on $[0,T]\times\mathbb{R}^n$.  
\end{assumption}

\begin{assumption}
 $\tilde K$ is independent of $q$
\end{assumption}
\begin{assumption}
For every $T>0$, $\nabla_q V$, $F$, and $\sigma$ are Lipschitz in $x$ uniformly in $t\in[0,T]$.  We also assume that $A$ and $\gamma$ are $C^2$, $\psi$ is $C^3$, and $\partial_t\psi$, $\partial_{q^i}\psi$, $\partial_{q^i}\partial_{q^j}\psi$, $\partial_t\partial_{q^i}\psi$, $\partial_t\partial_{q^j}\partial_{q^i}\psi$, $\partial_{q^l}\partial_{q^j}\partial_{q^i}\psi$, $\partial_t\gamma$, $\partial_{q^i} \gamma$, $\partial_t\partial_{q^j}\gamma$,  $\partial_{q^i}\partial_{q^j}\gamma$, $\partial_t A$, $\partial_{q^i} A$, $\partial_t \partial_{q^i}A$,  and $\partial_{q^i}\partial_{q^j} A$ are bounded on $[0,T]\times\mathbb{R}^{n}$ for every $T>0$.
\end{assumption}

\section{Linear Algebra Lemmas}

\setcounter{lemma}{0}
    \renewcommand{\thelemma}{\Alph{section}\arabic{lemma}}

For the benefit of the reader we collect some more or less well  linear algebra lemmas in this appendix.
\begin{lemma}\label{eig_bound_lemma1}
Let $A$ be an $n\times n$-real or complex matrix with symmetric part $A^s=\frac{1}{2}(A+A^*)$. If the eigenvalues of $A^s$ are bounded above (resp. below) by $\alpha$ then the real parts of the eigenvalues of $A$ are bounded above (resp. below) by $\alpha$.
\end{lemma}
\begin{proof}
Suppose $y$ is an eigenvector of $A$ with norm $1$ corresponding to the eigenvalue $\lambda$.  Let $A^a$ be the antisymmetric part of $A$. 
\begin{align}
\Re(y^*A^ay)=\Re(\overline{y^*A^ay})=\Re(y^*(A^a)^*y)=-\Re(y^*A^ay).
\end{align}
 Therefore $\Re(y^*A^ay)=0$ and
\begin{align}
\Re(\lambda)=\Re( y^*Ay)=\Re( y^*A^sy).
\end{align}
If the eigenvalues of $A^s$ are bounded above by $\alpha$ then
\begin{align}
\Re(\lambda)=\Re( y^*A^sy)\leq \alpha\|y\|^2=\alpha
\end{align}
and if they are bounded below by $\alpha$ then
\begin{align}
\Re(\lambda)=\Re( y^*A^sy)\geq \alpha\|y\|^2=\alpha.
\end{align}
\end{proof}

\begin{lemma}\label{eig_bound_lemma2}
Let $A$ be a positive definite $n\times n$-real matrix with eigenvalues bounded below by $\lambda>0$ and $B$ be an $n\times n$-real matrix whose symmetric part has eigenvalues  bounded below by $\gamma>0$.  Then the eigenvalues of $AB$ have real parts bounded below by $\gamma\lambda$.
\end{lemma}
\begin{proof}
$A$ is positive definite, so we can factor it as $A=DD^T$ where $D$ is a real valued, invertible, $n\times n$-matrix. $AB$ and the conjugation $D^{-1}ABD=D^TBD$ have the same eigenvalues.  The symmetric part of $D^TBD$ is $D^TB^sD$ and for any $y\in\mathbb{R}^n$,
\begin{align}
y^TD^TB^sDy\geq \gamma y^TD^TDy.
\end{align}
$D^TD$ is a positive definite matrix that has the same eigenvalues as $A=DD^T$ (both are equal to the squared singular values of $D$).  The eigenvalues of $A$ are bounded below by $\lambda$, so
\begin{align}
  y^TD^TDy\geq \lambda \|y\|^2.
\end{align}
Therefore the eigenvalues of the symmetric part of $D^TBD$ are bounded below by $\gamma\lambda$.  Hence, by Lemma \ref{eig_bound_lemma1}, the real parts of the eigenvalues of $D^TBD$, and hence of $AB$, are bounded below by $\gamma\lambda$.
\end{proof}

\begin{lemma}\label{A_inv_pos_def_lemma}
Let $A$ be an $n\times n$-real or complex matrix whose symmetric part has eigenvalues bounded below by $\lambda>0$. Then $A$ is invertible and the symmetric part of $A^{-1}$ has eigenvalues bounded below by $\lambda/\|A\|^2$.
\end{lemma}
\begin{proof}
Lemma \ref{eig_bound_lemma1} implies that $A$ is invertible.  Let $v\in V$ be non-zero.
\begin{align}
&\langle v,(A^{-1})^sv\rangle=\frac{1}{2}\Re(\langle v,A^{-1}v\rangle+\langle v,(A^*)^{-1}v\rangle)\\
=&\frac{1}{2}\Re(\langle AA^{-1}v,A^{-1}v\rangle+\langle A^*(A^*)^{-1}v,(A^*)^{-1}v\rangle)\notag\\
=&\frac{1}{2}(\Re(\langle A^{-1}v,(A^*)^s(A^{-1}v)\rangle)+\Re(\langle A^{-1}v,(A^*)^a(A^{-1}v)\rangle))\notag\\
&+\frac{1}{2}(\Re(\langle (A^*)^{-1}v,A^s((A^*)^{-1}v)\rangle)+\Re(\langle (A^*)^{-1}v,A^a((A^*)^{-1}v)\rangle))\notag\\
=&\frac{1}{2}(\langle A^{-1}v,A^s(A^{-1}v)\rangle+0)+\frac{1}{2}(\langle (A^*)^{-1}v,A^s((A^*)^{-1}v)\rangle+0)\notag\\
\geq&  \frac{\lambda}{2}\left(\langle v,(AA^*)^{-1}v\rangle+\langle v,  (A^*A)^{-1}v\rangle\right).\notag
\end{align}
Using the singular value decomposition of $A$ we see that both $(AA^*)^{-1}$ and $(A^*A)^{-1}$ are positive definite with eigenvalues bounded below by $1/\|A\|^2$.  The result follows.

\end{proof}

\section{Non-Explosion of Solutions}\label{app:no_explosions}
In the course of proving our main result, Theorem \ref{conv_thm}, we showed that the limiting process, $q_t$, exist for all $t\geq 0$ with probability one, at least under Assumptions \ref{assump1}-\ref{assump7}. Though we have assumed it to be the case throughout this paper, the same is not obvious for the family of  solutions, $x_t^\epsilon$, to the SDE \ref{Hamiltonian_SDE_q}-\ref{Hamiltonian_SDE_p}. However, existence for all $t\geq 0$  can be proven under a collection of assumptions that are very similar to our Assumptions \ref{assump1}-\ref{assump7} from the main text, as shown in the following lemma. We emphasize that in this appendix, we do not employ any of the Assumptions \ref{assump1}-\ref{assump7} per se.  The assumptions we do use are all listed below, in the statement of the lemma.

\setcounter{lemma}{0}
    \renewcommand{\thelemma}{\Alph{section}\arabic{lemma}}

\setcounter{corollary}{0}
    \renewcommand{\thecorollary}{\Alph{section}\arabic{corollary}}

\begin{lemma}\label{unique_lemma1}
Suppose:
\begin{enumerate}
\item The family of Hamiltonians have the form 
\begin{align}
H^\epsilon (t,x)=K(\epsilon,t,q,(p-\psi(t,q))/\sqrt{\epsilon})+V(t,q)
\end{align}
where $ K(\epsilon,t,q,z)$ is non-negative, $C^2$ in $(t,q,z)$ for every $\epsilon$,   $V$ is $C^2$, and $\psi$ is a $C^2$, $\mathbb{R}^n$-valued function.  
\item There exists $a\geq 0, b\geq 0$ such that $\tilde V(t,q)\equiv a+b\|q\|^2+V(t,q)$ is non-negative.
\item $\gamma(t,x)$, $\sigma(t,x)$, and $F(t,x)$ are continuous and  locally Lipschitz in $x$ with the Lipschitz constant uniform on compact time intervals.

\item $\sigma$ is bounded.

\item  The eigenvalues of $\gamma$ (which are real, since $\gamma$ is symmetric) are bounded below by some $\lambda>0$.
\item  For every $\epsilon\in (0,\epsilon_0]$, $t\geq 0$ there exists $c>0$, $N>0$, $\eta> 0$ such that 
\begin{align}
K(\epsilon,t,q,z)\geq c\|z\|^{2\eta}
\end{align}
for all $\|z\|\geq N$.
\item There exist $C>0$ and $M>0$ such that
\begin{align}
\left(\sum_{i,j}|\partial_{z_i}\partial_{z_j}K(\epsilon,t,q,z)|^2\right)^{1/2}\leq M+CK(\epsilon,t,q,z),
\end{align}

\begin{align}\label{H_dot_assump}
|\partial_t K(\epsilon,t,q,z)+\partial_t V(t,q)|\leq M+C\left(\|q\|^2+\tilde V(t,q)+K(\epsilon,t,q,z)\right),
\end{align}
and
\begin{align}\label{psi_F_bound}
\|-\partial_t\psi(t,q)+F(t,x)\|^2\leq M+C\left(\|q\|^2+\tilde V(t,q)\right).
\end{align}
\end{enumerate}

Then the maximal solution, $x_t^\epsilon$, to the SDE \ref{Hamiltonian_SDE_q}-\ref{Hamiltonian_SDE_p} is unique a.s. and a.s. exist for all $t\geq 0$.

\end{lemma}
\begin{proof}
Fix $\epsilon>0$. The assumptions imply that an a.s. unique, maximal solution $x_t^\epsilon$ exists up to explosion time $e^\epsilon$ (see Section 3.4 in \cite{khasminskii2011stochastic}). Non-explosion of $x^\epsilon_t$ (i.e. $e^\epsilon=\infty$ a.s.) will follow from the existence of a Lyapunov function (see Theorem 3.5 in \cite{khasminskii2011stochastic}), a non-negative $C^{2}$ function, $U(t,x)$, that satisfies:
\begin{enumerate}
\item For any $t\geq 0$, $\lim_{x\rightarrow\infty}U(t,x)=\infty$,
\item\begin{align}
L[U](t,x)\leq \tilde M+\tilde CU(t,x)
\end{align}
for some $\tilde M\geq 0$, $\tilde C> 0$, where $L$ is the time-dependent generator
\begin{align}
L[U](t,x)=&\partial_tU(t,x)+\nabla_p H^\epsilon(t,x)\cdot \nabla_qU(t,x)\\
&+(-\gamma(t,x)\nabla_p H^\epsilon(t,x)-\nabla_q H^\epsilon(t,x)+F(t,x))\cdot \nabla_p U(t,x)\notag\\
&+\frac{1}{2}\Sigma_{ij}(t,x)\partial_{p_i}\partial_{p_j} U(t,x)\notag
\end{align}
and $\Sigma_{ij}=\sum_\rho\sigma_{i\rho}\sigma_{j\rho}$.

We note that to connect with the result as stated in \cite{khasminskii2011stochastic}, one needs to work with $\tilde M/\tilde C+U$ in place of $U$, but we find the above formulation more convenient here.
\end{enumerate}

Fix $\epsilon>0$. As our candidate Lyapunov function, we  define
\begin{align}
U(t,x)\equiv \|q\|^2+\tilde V(t,q)+K^\epsilon(t,x)=a+(1+b)\|q\|^2+H^\epsilon(t,x).
\end{align}
By assumption, $\tilde V(t,q)$ and $K^\epsilon(t,x)$ are non-negative and $C^2$, therefore $U$ is as well.

Fix $t\geq 0$. By assumption there exists $c>0$, $N>0$, $\eta> 0$ such that 
\begin{align}
K(\epsilon,t,q,z)\geq c\|z\|^{2\eta}
\end{align}
for all $\|z\|\geq N$.

Given $R>0$, let 
\begin{equation}
\max\{\|q\|,\|p\|\}\geq \max\left\{R^{1/2},\sup_{\|q\|\leq R^{1/2}}\|\psi(t,q)\|+\epsilon^{1/2}\left(N+(R/c)^{1/2\eta}\right)\right\}.
\end{equation}
If $\|q\|\geq R^{1/2}$ then $U(t,x)\geq \|q\|^2\geq R$.  If $\|q\|<R^{1/2}$ then 
\begin{align}
\|p-\psi(t,q)\|/\sqrt{\epsilon}\geq \left(\|p\|-\sup_{\|q\|\leq R^{1/2}} \|\psi(t,q)\|\right)/\sqrt{\epsilon}\geq N.
\end{align}
Hence
\begin{align}
U(t,x)\geq& K^\epsilon(t,x)\geq c\|p-\psi(t,q)\|^{2\eta}/\epsilon^\eta\geq \frac{c}{ \epsilon^\eta}\left(\|p\|-\sup_{\|q\|\leq R^{1/2}}\|\psi(t,q)\|\right)^{2\eta}\\
\geq & \frac{c}{ \epsilon^\eta}\left(\epsilon^{1 \over 2}\left({R \over c}\right)^{1 \over 2\eta}\right)^{2\eta} = R.\notag
\end{align}
Therefore $U(t,x)\rightarrow\infty$ as $x\rightarrow\infty$.

Using the inequality $\alpha\beta\leq \frac{1}{2}(\delta \alpha^2+\frac{1}{\delta}\beta^2)$ for any $\alpha>0$, $\beta>0$, $\delta>0$, we can compute
\begin{align}
L[U](t,x)=&\partial_tH^\epsilon(t,x)+\nabla_p H^\epsilon(t,x)\cdot \nabla_qH^\epsilon(t,x)\\
&+(-\gamma(t,x)\nabla_p H^\epsilon(t,x)-\nabla_q H^\epsilon(t,x)+F(t,x))\cdot \nabla_p H^\epsilon(t,x)\notag\\
&+\frac{1}{2}\Sigma_{ij}(t,x)\partial_{p_i}\partial_{p_j} H^\epsilon(t,x)+2(1+b)\nabla_pH^\epsilon(t,x)\cdot q\notag
\end{align}
\begin{align}
\leq&(\partial_tK)^\epsilon(t,x)+\partial_t V(t,q)-\frac{\lambda}{\epsilon} \|(\nabla_z K)^\epsilon(t,x)\|^2 \notag\\
&+\frac{1}{\sqrt{\epsilon}}\left( 2(1+b)q+F(t,x)- \partial_t\psi(t,q)\right)\cdot (\nabla_z K)^\epsilon(t,x)\notag\\
&+\frac{1}{2\epsilon}\Sigma_{ij}(t,x)(\partial_{z_i}\partial_{z_j} K)^\epsilon(t,x)\notag\\
\leq& M+CU(t,x)-\frac{\lambda}{\epsilon} \|(\nabla_z K)^\epsilon(t,x)\|^2 +\frac{1}{2\epsilon}\|\Sigma\|_{F,\infty}(M+CK^\epsilon(t,x))\notag\\
&+\frac{1}{4\lambda}\| 2(1+b)q+F(t,x)- \partial_t\psi(t,q)\|^2+\frac{\lambda}{\epsilon}\| (\nabla_z K)^\epsilon(t,x)\|^2\notag\\
\leq& M+CU(t,x)+\frac{\lambda}{\epsilon}\| (\nabla_z K)^\epsilon(t,x)\|^2\notag\\
&+\frac{1}{2\lambda}( 4(1+b)^2\|q\|^2+M+C(\|q\|^2+\tilde V(t,q))).\notag
\end{align}
The right hand side is bounded by $\tilde M+\tilde C U$ for some $\tilde M\geq 0$, $\tilde C>0$.  This completes the proof that $U$ is a Lyapunov function and allows us to conclude that $e^\epsilon=\infty$ a.s. i.e. the solution $x^\epsilon_t$ exists for all $t\geq 0$ a.s.
\end{proof}

A similar result holds for proving non-explosion of the limiting equation, \req{limit_eq1}, under assumptions weaker than Assumptions \ref{assump1}-\ref{assump7}. The following lemma can be proven by using  the Lyapunov function $U(t,q)\equiv a+\left(1+b\right)\|q\|^2+V(t,q)$ and  Lemma \ref{A_inv_pos_def_lemma}. The proof closely follows that of the previous lemma and so we leave the details to the reader.
\begin{lemma}\label{unique_lemma2}
Consider an SDE on $\mathbb{R}^n$ of the form
\begin{align}\label{model_limiting_SDE}
dq_t=&\tau^{-1}(t,q_t)(-\nabla_{q}V(t,q_t)+\tilde F(t,q_t))dt+\tilde\sigma(t,q_t) dW_t
\end{align}
where
\begin{enumerate}
\item $V$ is $C^2$ and there exists $a>0$, $b>0$ such that $\tilde V(t,q)\equiv a+b\|q\|^2+V(t,q)$ is non-negative,
 \item $\tilde\sigma(t,q)$, and $\tilde F(t,q)$ are continuous and  locally Lipschitz in $q$ with the Lipschitz constant uniform on compact time intervals,
\item $\tilde\sigma$ is bounded,
\item  $\tau$ is $C^1$, bounded, and its symmetric part has eigenvalues bounded below by $\lambda>0$,
\item there exists $M>0$ and $C>0$ such that
\begin{align}
\left(\sum_{i,j}|\partial_{q^i}\partial_{q^j} V(t,q)|^2\right)^{1/2}\leq M+C\left(\|q\|^2+\tilde V(t,q)\right),
\end{align}
\begin{align}
|\partial_t V(t,q)|\leq  M+C\left(\|q\|^2+\tilde V(t,q)\right),
\end{align}
and
\begin{align}
\|\tilde F(t,q)\|^2\leq   M+C\left(\|q\|^2+\tilde V(t,q)\right).
\end{align}
\end{enumerate}

Then the maximal solution to the SDE  \ref{model_limiting_SDE} is unique a.s. and a.s. exist for all $t\geq 0$.

\end{lemma}

\subsection*{Acknowledgments}

J.W. was partially supported by NSF grants DMS 131271 and DMS 1615045.\\

\bibliographystyle{ieeetr}
\bibliography{refs}

\begin{thebibliography}{10}

\bibitem{smoluchowski1916drei}
M.~Smoluchowski, ``Drei vortrage uber diffusion, brownsche bewegung und
  koagulation von kolloidteilchen,'' {\em Zeitschrift fur Physik}, vol.~17,
  pp.~557--585, 1916.

\bibitem{KRAMERS1940284}
H.~Kramers, ``Brownian motion in a field of force and the diffusion model of
  chemical reactions,'' {\em Physica}, vol.~7, no.~4, pp.~284 -- 304, 1940.

\bibitem{doi:10.1137/S1540345903421076}
G.~A. Pavliotis and A.~M. Stuart, ``White noise limits for inertial particles
  in a random field,'' {\em Multiscale Modeling \& Simulation}, vol.~1, no.~4,
  pp.~527--553, 2003.

\bibitem{Chevalier2008}
C.~Chevalier and F.~Debbasch, ``Relativistic diffusions: A unifying approach,''
  {\em Journal of Mathematical Physics}, vol.~49, no.~4, 2008.

\bibitem{bailleul2010stochastic}
I.~Bailleul, ``A stochastic approach to relativistic diffusions,'' in {\em
  Annales de l'institut Henri Poincar{\'e} (B)}, vol.~46, pp.~760--795, 2010.

\bibitem{pinsky1976isotropic}
M.~A. Pinsky, ``Isotropic transport process on a riemannian manifold,'' {\em
  Transactions of the American Mathematical Society}, vol.~218, pp.~353--360,
  1976.

\bibitem{pinsky1981homogenization}
M.~A. Pinsky, ``Homogenization in stochastic differential geometry,'' {\em
  Publications of the Research Institute for Mathematical Sciences}, vol.~17,
  no.~1, pp.~235--244, 1981.

\bibitem{Jorgensen1978}
E.~J{\o}rgensen, ``Construction of the brownian motion and the
  ornstein-uhlenbeck process in a riemannian manifold on basis of the
  gangolli-mc.kean injection scheme,'' {\em Zeitschrift f{\"u}r
  Wahrscheinlichkeitstheorie und Verwandte Gebiete}, vol.~44, no.~1,
  pp.~71--87, 1978.

\bibitem{dowell1980differentiable}
R.~M. Dowell, {\em Differentiable approximations to Brownian motion on
  manifolds}.
\newblock PhD thesis, University of Warwick, 1980.

\bibitem{XueMei2014}
X.-M. {Li}, ``{Random Perturbation to the Geodesic Equation},'' {\em Ann.
  Prob.}, vol.~44, no.~1, pp.~544--566, 2016.

\bibitem{angst2015kinetic}
J.~Angst, I.~Bailleul, and C.~Tardif, ``Kinetic brownian motion on riemannian
  manifolds,'' {\em arXiv preprint arXiv:1501.03679}, 2015.

\bibitem{bismut2005hypoelliptic}
J.-M. Bismut, ``The hypoelliptic laplacian on the cotangent bundle,'' {\em
  Journal of the American Mathematical Society}, vol.~18, no.~2, pp.~379--476,
  2005.

\bibitem{bismut2015}
J.-M. Bismut, ``Hypoelliptic laplacian and probability,'' {\em J. Math. Soc.
  Japan}, vol.~67, pp.~1317--1357, 10 2015.

\bibitem{Nelson1967}
E.~Nelson, {\em Dynamical Theories of Brownian Motion}.
\newblock Mathematical Notes - Princeton University Press, Princeton University
  Press, 1967.

\bibitem{pavliotis2008multiscale}
G.~Pavliotis and A.~Stuart, {\em Multiscale Methods: Averaging and
  Homogenization}.
\newblock Texts in Applied Mathematics, Springer New York, 2008.

\bibitem{PhysRevA.25.1130}
P.~H{\"a}nggi, ``Nonlinear fluctuations: The problem of deterministic limit and
  reconstruction of stochastic dynamics,'' {\em Phys. Rev. A}, vol.~25,
  pp.~1130--1136, Feb 1982.

\bibitem{volpe2010influence}
G.~Volpe, L.~Helden, T.~Brettschneider, J.~Wehr, and C.~Bechinger, ``Influence
  of noise on force measurements,'' {\em Physical review letters}, vol.~104,
  no.~17, p.~170602, 2010.

\bibitem{Sancho1982}
J.~M. Sancho, M.~S. Miguel, and D.~D{\"u}rr, ``Adiabatic elimination for
  systems of brownian particles with nonconstant damping coefficients,'' {\em
  Journal of Statistical Physics}, vol.~28, no.~2, pp.~291--305, 1982.

\bibitem{Hottovy2014}
S.~Hottovy, A.~McDaniel, G.~Volpe, and J.~Wehr, ``{The Smoluchowski-Kramers
  Limit of Stochastic Differential Equations with Arbitrary State-Dependent
  Friction},'' {\em Communications in Mathematical Physics}, vol.~336, no.~3,
  pp.~1259--1283, 2014.

\bibitem{herzog2015small}
D.~P. Herzog, S.~Hottovy, and G.~Volpe, ``{The small-mass limit for Langevin
  dynamics with unbounded coefficients and positive friction},'' {\em Journal
  of Statistical Physics}, vol.~163, no.~3, pp.~659--673, 2016.

\bibitem{particle_manifold_paper}
J.~{Birrell}, S.~{Hottovy}, G.~{Volpe}, and J.~{Wehr}, ``{Small Mass Limit of a
  Langevin Equation on a Manifold},'' {\em ArXiv e-prints}, Apr. 2016.

\bibitem{Chetrite2008}
R.~Chetrite and K.~Gaw\c{e}dzki, ``Fluctuation relations for diffusion
  processes,'' {\em Communications in Mathematical Physics}, vol.~282, no.~2,
  pp.~469--518, 2008.

\bibitem{gawedzki2013fluctuation}
K.~Gaw\c{e}dzki, ``Fluctuation relations in stochastic thermodynamics,'' {\em
  arXiv preprint arXiv:1308.1518}, 2013.

\bibitem{karatzas2014brownian}
I.~Karatzas and S.~Shreve, {\em Brownian Motion and Stochastic Calculus}.
\newblock Graduate Texts in Mathematics, Springer New York, 2014.

\bibitem{KHOA1992102}
D.~T. Khoa, N.~Ohtsuka, M.~Matin, A.~Faessler, S.~Huang, E.~Lehmann, and R.~K.
  Puri, ``In-medium effects in the description of heavy-ion collisions with
  realistic nn interactions,'' {\em Nuclear Physics A}, vol.~548, no.~1,
  pp.~102 -- 130, 1992.

\bibitem{ortega2013matrix}
J.~Ortega, {\em Matrix Theory: A Second Course}.
\newblock University Series in Mathematics, Springer US, 2013.

\bibitem{ChucksVolumePaper}
J.~Birrell and J.~Wehr, ``A homogenization theorem for langevin systems with an
  application to hamiltonian dynamics,'' {\em arXiv preprint arXiv:1707.02884},
  2017.

\bibitem{exp_deriv}
R.~M. Wilcox, ``Exponential operators and parameter differentiation in quantum
  physics,'' {\em Journal of Mathematical Physics}, vol.~8, no.~4, 1967.

\bibitem{khasminskii2011stochastic}
R.~Khasminskii, {\em Stochastic stability of differential equations}, vol.~66.
\newblock Springer Science \& Business Media, 2011.

\end{thebibliography}

\end{document}